\documentclass[11pt]{article}
\usepackage{fullpage,times,amsmath,amssymb,amsthm,graphicx}

\graphicspath{{./Figures/}}

\usepackage[numbers,sort&compress]{natbib}
\usepackage[T1]{fontenc}
\usepackage{xcolor}

\usepackage{mathtools,bm}

\usepackage{hyperref}
\usepackage{cleveref}

\DeclareMathOperator*{\polylog}{polylog}

\newcommand{\LE}{\textit{LE}}
\newcommand{\VDLE}{\textit{Decomp}}
\newcommand{\DT}{\textit{DT}}
\newcommand{\EMST}{\textit{EMST}}
\newcommand{\CH}{\textit{CH}}
\newcommand{\RNG}{\textit{RNG}}

\def\RR{\ensuremath{\mathbb R}}

\newcommand{\DEG}{\textrm{deg}}
\newcommand{\RAY}[1]{\overrightarrow{#1}}

\newcommand{\Dwrong}{D}
\newcommand{\dCross}{d_{\mbox{\scriptsize\rm cross}}}
\newcommand{\dVio}{d_{\mbox{\scriptsize\rm vio}}}
\newcommand{\Dflip}{D_{\mbox{\scriptsize\rm flip}}}
\newcommand{\Dcross}{D_{\mbox{\scriptsize\rm cross}}}
\newcommand{\Dvio}{D_{\mbox{\scriptsize\rm vio}}}
\newcommand{\dCrossMST}{d_{\mbox{\scriptsize\rm cross-emst}}}
\newcommand{\Dlocal}{D_{\mbox{\scriptsize\rm local}}}
\newcommand{\Dmst}{D_{\mbox{\scriptsize\rm emst}}}
\newcommand{\Drng}{D_{\mbox{\scriptsize\rm rng}}}
\newcommand{\Dtree}{D_T}

\newtheorem{theorem}{Theorem}[section]
\newtheorem{lemma}[theorem]{Lemma}
\newtheorem{corollary}[theorem]{Corollary}
\crefname{theorem}{Theorem}{Theorems}
\crefname{lemma}{Lemma}{Lemmas}
\crefname{corollary}{Corollary}{Corollaries}
\crefname{section}{Section}{Sections}

\definecolor{defblueee}{rgb}{0.1,0.4,0.6} 
\def\DEF#1{\textbf{\emph{\textcolor{defblueee}{#1}}}}

\title{Delaunay Triangulations with Predictions}

\author{Sergio Cabello\thanks{Faculty of Mathematics and Physics, University of Ljubljana, Slovenia and Institute of Mathematics, Physics and Mechanics, Slovenia, Email: \texttt{sergio.cabello@fmf.uni-lj.si}
} \and 
Timothy M. Chan\thanks{Siebel School of Computing and Data Science, University of Illinois at Urbana-Champaign, USA, Email: \texttt{tmc@illinois.edu}} \and Panos Giannopoulos\thanks{Department of Computer Science, City St George's, University of London, UK, Email: \texttt{Panos.Giannopoulos@city.ac.uk}}
}
\begin{document}

\maketitle

\thispagestyle{empty}

\begin{abstract}
We investigate \emph{algorithms with predictions} in computational geometry, 
specifically focusing on the basic problem of computing 2D Delaunay triangulations.
Given a set $P$ of $n$ points in the plane and a triangulation $G$ that serves
as a ``prediction'' of the Delaunay triangulation, we would like to use $G$ to compute
the correct Delaunay triangulation $\textit{DT}(P)$ more quickly when $G$ is ``close'' to $\textit{DT}(P)$.  We obtain a variety of results of this type, under different deterministic and probabilistic settings, including the following:
\begin{enumerate}
\item Define $D$ to be the number of edges in $G$ that are not in $\textit{DT}(P)$.
We present a deterministic algorithm to compute $\textit{DT}(P)$ from $G$ 
in $O(n + D\log^3 n)$ time, and a randomized algorithm in $O(n+D\log n)$ expected time,
the latter of which is optimal in terms of $D$.
\item Let $R$ be a random subset of the edges of $\textit{DT}(P)$, where each edge is chosen independently with probability $\rho$.
Suppose $G$ is any triangulation of $P$ that contains $R$.
We present an algorithm to compute $\textit{DT}(P)$ from $G$ 
in $O(n\log\log n + n\log(1/\rho))$ time with high probability.
\item Define $d_{\mbox{\scriptsize\rm vio}}$ to be the maximum number of points of $P$ strictly inside the circumcircle of a triangle in $G$ (the number is 0 if $G$ is equal to $\textit{DT}(P)$).
We present a deterministic algorithm to compute $\textit{DT}(P)$ from $G$  
in $O(n\log^*n + n\log d_{\mbox{\scriptsize\rm vio}})$ time. 
\end{enumerate}
We also obtain results in similar settings for related problems such as 2D Euclidean minimum spanning trees, and hope that our work will open up a fruitful line of future
research.
\end{abstract}

\newpage
\setcounter{page}{1}

\section{Introduction}

\subsection{Towards computational geometry with predictions}

\emph{Algorithms with predictions} have received considerable attention in recent years.
Given a problem instance and a prediction for the solution, we would like to devise an algorithm that solves the problem in such a way that the algorithm performs ``better'' when the prediction is ``close'' to the correct solution.  ``Better'' may be in terms of running time, or approximation factor (for optimization problems), or competitive ratio (for on-line problems).  Different measures of ``closeness''  may be used depending on the problem, or alternatively we may assume that predictions obey some probabilistic models. The paradigm has been applied to problems in diverse areas, such as scheduling, classical graph problems, resource allocation, mechanism design, facility location, and graph searching to name a few; see the survey in~\cite{MitzenmacherV_BeyondW_CAA20} and the regularly updated online list of papers in~\cite{Alg_with_pred}.

In contrast, work on algorithms with predictions in \emph{computational geometry}
have been few and far between. The goal of this paper is to promote this direction of research.  Primarily, we concentrate on the use of predictions to improve the running time of geometric algorithms.

As many fundamental problems in classical computational geometry may be viewed as two- or higher-dimensional generalizations of standard one-dimensional problems like sorting and searching, we start our discussion with the 1D sorting problem.  Here, algorithms with predictions are related to the historically well-explored topic of \emph{adaptive sorting} algorithms~\cite{Estivill-CastroW92}---if we are given a sequence that is already close to sorted, we can sort faster.
There are different ways to measure distance to sortedness.  For example:
\begin{itemize}
\item Consider $D_{\textrm{inv}}$, the number of inversions in the given sequence (i.e., the number of out-of-order pairs).  In terms of this parameter, good old 
insertion-sort runs in $O(n + D_{\textrm{inv}})$ time and already beats 
$O(n\log n)$ algorithms when $D_{\textrm{inv}}=O(n)$. Better
algorithms are known running in $\Theta(n\log(1+D_{\textrm{inv}}/n))$ time, which is optimal
in comparison-based models.
\item Or consider $D_{\textrm{rem}}$, the minimum number of elements to remove (the ``wrong'' elements) so that the remaining list is sorted.  Here, the optimal time bound is known to be $\Theta(n+D_{\textrm{rem}}\log n)$. 
\item Or consider $D_{\textrm{runs}}$, the number of runs (i.e., the number of increasing contiguous subsequence that the input can be partitioned into).
Here, the optimal time bound is known to be $\Theta(n\log D_{\textrm{runs}})$.  Unfortunately, even when $D_{\textrm{runs}}$ is as low as
$n^{0.1}$, it is not possible to beat $n\log n$ under this distance measure.
\end{itemize}
The survey in~\cite{Estivill-CastroW92} gives an extensive background and a (long) list of other parameters that have been considered in the algorithms literature.  Adaptive sorting  has continued to be studied (some random recent examples include
\cite{BarbayN13, MunroW18, AugerJNP18}), and some recent papers~\cite{BaiC23,LuRSZ21} examine newer models for sorting with
predictions.


\subsection{The 2D Delaunay triangulation problem}

Can we obtain similar results for geometric problems in 2D?  In this paper, we pick arguably one of the most basic problems in computational geometry to start this study: namely, the computation of the \emph{Delaunay triangulation}, denoted $\DT(P)$,
of a set $P$ of $n$ points in the plane~\cite{BergCKO08}.  One way to define $\DT(P)$ is via the \emph{empty circle property}: three points $p,p',p''\in P$ form a Delaunay triangle iff
the circumcircle through $p,p',p''$ does not have any points of $P$ strictly inside.  (Throughout this paper, we assume for simplicity that the input is in general position, e.g., no 4 input points are co-circular.)

Delaunay triangulations are combinatorially equivalent to \emph{Voronoi diagrams} (under the Euclidean metric) by duality.  (It is a little more convenient to work with Delaunay triangulations than Voronoi diagrams when specifying the prediction $G$ or defining closeness, and so all our results will be stated in terms of the former.)
Delaunay triangulations and Voronoi diagrams are among the most ubiquitous geometric structures, with countless applications, and lie at the very core of computational geometry. 

In the prediction setting, we are given not just the input point set $P$, but also
a triangulation $G$ of $P$ (i.e., a decomposition of the convex hull of $P$ into nonoverlapping triangles, using all the 
points of $P$ as vertices, but no Steiner points). The goal is to devise an algorithm to compute $\DT(P)$ faster, when the ``distance'' between $G$ and $\DT(P)$ is small.
In the traditional setting, several standard textbook algorithms are known for computing the Delaunay triangulation in $O(n\log n)$ time, e.g., via divide-and-conquer, plane sweep, and randomized incremental construction~\cite{BoissonatY98, BergCKO08,AurenhammerK00,Fortune04}, which is optimal in the worst case.  Without closeness, an arbitrary triangulation $G$ of $P$ still provides some extra information about the input, but not enough to beat $n\log n$: it is not difficult to see that
$\Omega(n\log n)$ remains a lower bound by reduction from 1D sorting (see Figure~\ref{fig:reduce:sort}).  We want to do better than $n\log n$ when $G$ is close.
As we will see, designing faster adaptive algorithms for 2D Delaunay triangulations  is  technically more difficult, and so more interesting, than for 1D sorting.

\begin{figure}\centering
		\includegraphics[scale=1,page=1]{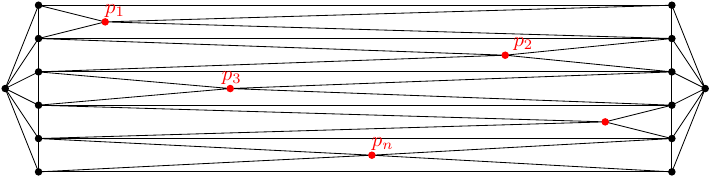}
		\caption{Sorting the $x$-coordinates of $p_1,\ldots,p_n$ reduces to computing the Delaunay triangulation of the $O(n)$ vertices of the triangulation shown (assuming a rescaling to compress the $y$-coordinates).}
		\label{fig:reduce:sort}
	\end{figure}

\subsection{Two specific questions}

\begin{table}\centering
\begin{tabular}{|l|l|}\hline
parameter & definition\\\hline\hline
$\Dwrong$ & number of edges in $G$ that are not in $\DT(P)$\\
& (or equivalently, number of edges in $\DT(P)$ that are not in $G$)\\[1ex]\hline
$\Dflip$ & minimum number of \emph{edge-flips} to transform $G$ into $\DT(P)$\\
& (in a flip, an edge $pp'$ incident to triangles $\triangle pp'p''$ and $\triangle pp'p'''$ is replaced\\
& by a new edge $p''p'''$, assuming that $pp''p'p'''$ is a convex quadrilateral)\\[1ex]\hline
$\Dcross$ & number of edge crossings between $G$ and $\DT(P)$\\[1ex]\hline
$\dCross$ & maximum number of edges of $\DT(P)$ crossed by an edge of $G$\\[1ex]\hline
$\Dvio$ & number of \emph{violations}, i.e., $(p''',\triangle pp'p'')$ with $p'''\in P$ and $\triangle pp'p''$ in $\DT(P)$\\
& such that $p'''$ lies strictly inside the circumcircle through $p,p',p''$\\[1ex]\hline
$\dVio$ & maximum number of points of $P$ strictly inside\\
& the circumcircle through a triangle in $G$\\[1ex]\hline
$\Dlocal$ & number of \emph{non-locally-Delaunay} edges in $G$\\
& (i.e., edges $pp'$ incident to triangles $\triangle pp'p''$ and $\triangle pp'p'''$ in $G$ such that\\
& $p'''$ is inside the circumcircle through $p,p',p''$)\\\hline 
\end{tabular}
\caption{Different measures of closeness between a given triangulation $G$ and the
Delaunay triangulation $\DT(P)$.}\label{tbl:def}
\end{table}

There are a number of natural parameters to measure distance to the Delaunay triangulation, as we define in Table~\ref{tbl:def}.  Some of these definitions (e.g., $\Dwrong,\Dflip,\Dcross$) may be used to compare any two triangulations, while others (e.g., $\Dvio,\Dlocal$) are specific to the Delaunay triangulation. For all these parameters, the value is 0 if and only if $G=\DT(P)$.  Some relationships between these parameters are obvious, e.g., $\Dcross\le O(n\dCross)$ and $\Dvio\le O(n\dVio)$, and it is known that
\[ \Dlocal\le \Dwrong\le\Dflip\le \Dcross.
\]
The inequality $\Dlocal\le\Dwrong$ holds because all non-locally Delaunay edges can't be actual Delaunay edges, while the last inequality $\Dflip\le\Dcross$ is nontrivial and was proved 
by Hanke et al.~\cite{HankeOS96}.

There are examples when $\Dlocal$ is a constant and $\Dwrong$ is $\Omega(n)$.
In a very vague sense, $\Dlocal$ is more similar to $D_{\textrm{runs}}$ (the number of runs) for sorting, while
$\Dwrong$ is more similar to $D_{\textrm{rem}}$ and $\Dflip$ and $\Dcross$ are more similar to $D_{\textrm{inv}}$.   Indeed, it is not difficult to prove an $\Omega(n\log\Dlocal)$ lower bound for computing the Delaunay triangulation in terms of $\Dlocal$, by reduction from sorting with $\Theta(\Dlocal)$ runs, or equivalently, merging $\Theta(\Dlocal)$ sorted lists (see Figure~\ref{fig:reduce:runs}).  So, $\Dlocal$ may not be an ideal parameter for the problem (even when $\Dlocal$ is as small as $n^{0.1}$, we can't hope to beat $n\log n$).

\begin{figure}\centering
		\includegraphics[scale=1.1,page=2]{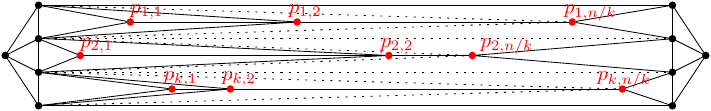}
		\caption{Merging $k$ $x$-sorted lists $\langle p_{1,1},\ldots,p_{1,n/k}\rangle$, \ldots,
		   $\langle p_{k,1},\ldots,p_{k,n/k}\rangle$ reduces to computing the Delaunay triangulation of the $O(n)$ vertices of the triangulation shown (assuming a rescaling to compress the $y$-coordinates).  This triangulation has $\Dlocal=\Theta(k)$ non-locally-Delaunay edges shown in dotted lines.}
		\label{fig:reduce:runs}
	\end{figure}

Arguably the simplest and most natural of all these parameters is $\Dwrong$, the number of ``wrong'' edges in $G$ that are not in $\DT(P)$ (which is the same as the number of edges in $\DT(P)$ that are in $G$, since all triangulations of the same point set have the same number of edges). To initiate the study of algorithms with predictions for the Delaunay triangulation, we thus begin with the following simple question:
\begin{center}
\it Question 1: what is the best 2D Delaunay triangulation algorithm with running time analyzed as a function of $n$ and the parameter $\Dwrong$?
\end{center}
This simple question already turns out to be quite challenging. An upper bound on time complexity in terms of $\Dwrong$ would automatically imply an upper bound in terms of $\Dflip$ and $\Dcross$ by the above inequalities.  We can ask a similar question for other parameters such as
$\dCross$ and $\dVio$.

Although Question 1 provides one way to model predictions, often errors are probabilistic in nature.  Thus, for various problems, algorithms with predictions have sometimes been studied under probabilistic models. 
To define a simple probabilistic model of prediction for the Delaunay triangulation, one could just consider tossing an independent coin to determine whether an edge of $\DT(P)$ should be placed in $G$ (similar to a model studied by Cohen-Addad et al.~\cite{Cohen-Addadd0LP24} for the max-cut problem). But there is one subtlety for our problem:
the input $G$ is supposed to be a triangulation, containing wrong edges (but not containing self-crossings).
We propose the following 2-stage model: first, each edge of $\DT(P)$ is placed in a set $R$ independently with a given probability $\rho$; second, $R$ is completed to a triangulation of $P$, but we make no assumption about how this second step is done, i.e., $R$ can be triangulated \emph{adversarially} (this makes results in this model a bit more general).

\begin{center}
\it Question 2: under the above probabilistic model, what is the best 2D Delaunay triangulation algorithm with running time analyzed as a function of $n$ and $\rho$?
\end{center}
In this model, the expected number of ``wrong'' edges $\Dwrong$ is upper-bounded by  $(1-\rho)m$, where $m\le 3n-6$ is the number of edges in any triangulation of $P$.  The hope is that the problem may be solved faster under the probabilistic model than by applying a worst-case algorithm in terms of the parameter $\Dwrong$.

We could consider other probabilistic models too.  For example, one variant is to pick a random edge in $\DT(P)$, perform a flip if its two incident triangles form a convex quadrilateral (or do nothing otherwise), and repeat for $(1-\rho)m$ steps.
In this model, the flip distance $\Dflip$ is guaranteed to be at most $(1-\rho)m$; in contrast, the above model allows for potentially larger flip distances and so is in some sense stronger.

\subsection{Our contributions}

We obtain a number of results addressing Questions 1 and 2 and related questions:

\begin{enumerate}
\item We present several different algorithms for computing $\DT(P)$ from $G$ in time sensitive to the parameter $\Dwrong$. One deterministic algorithm runs in 
$O(n+\Dwrong\log^4n)$ time, 
and another deterministic algorithm runs in 
$O(n+\Dwrong\log^3n)$ time.
A randomized algorithm runs in $O((n+\Dwrong\log n)\log^*n)$ time, where
$\log^*$ is the (very slow-growing) iterated logarithmic function.
Our final randomized algorithm eliminates the $\log^* n$ factor and runs in 
$O(n+\Dwrong\log n)$ expected time, which is optimal  since an $\Omega(\max\{n,\Dwrong\log\Dwrong\})=\Omega(n+\Dwrong\log n)$ lower bound follows by reduction from sorting
$\Dwrong$ numbers (we can just embed the configuration in Figure~\ref{fig:reduce:sort} on $\Theta(\Dwrong)$ points into a larger Delaunay triangulation).  Thus, we have completely answered Question~1 (modulo the use of randomization).  All of our algorithms do not require knowing  the value of $\Dwrong$ a priori.

It is interesting to note that our optimal randomized algorithm runs in \emph{linear} time even when the number $\Dwrong$ of ``wrong'' edges is as large as $n/\log n$.

An immediate consequence is that the problem can be solved in $O(n+\Dflip\log n)$ expected time---we get linear running time as long as the flip distance is at most $n/\log n$.

\item On the probabilistic model in Question 2, we obtain another algorithm for computing $\DT(P)$ from $G$ in $O(n\log\log n + n\log(1/\rho))$ time with high probability (the algorithm is always correct).  For any constant $\rho$, the running time is thus almost  linear except for the small $\log\log n$ factor. This is notable when compared against our $O(n+\Dwrong\log n)$ algorithm,  since in this probabilistic setting, $\Dwrong$ is expected to be near $(1-\rho)m$.  For example, if we set $\rho=0.1$ 
(or even $\rho=1/\log n$), we can still substantially beat $n\log n$ even though a majority of the edges are wrong!

We can also get a similar result for the random flip model mentioned earlier.

\item We also obtain results sensitive to other parameters.
For example, in terms of $\dVio$, we describe a deterministic algorithm for computing $\DT(P)$ from $G$ in $O(n\log^*n + n\log\dVio)$ time.  Note that there are examples where $\dVio=1$ but $\Dwrong=\Omega(n)$ (i.e., a fraction of the edges are wrong!), so our $\Dwrong$-sensitive algorithms can't be used directly.
The parameter $\dVio$ is related
to a notion called \emph{order-$k$ Delaunay triangulation}%
\footnote{
The name is somewhat misleading as it is not equivalent to the \emph{order-$k$ Voronoi diagram} in the dual. The set of all candidate triangles for order-$k$ Delaunay triangulations does correspond to the vertices of the order-$k$ Voronoi diagram, but an order-$k$ Delaunay triangulation is not unique, comprising of a subset of the candidate triangles.
} 
in the literature~\cite{GudmundssonHK02} (with motivation very different from ours):
a triangulation $G$ is an ``order-$k$ Delaunay triangulation'' precisely when $\dVio\le k$.
Our result implies the following interesting consequence: given an order-$O(1)$ Delaunay triangulation, we can recover the actual Delaunay triangulation in almost linear ($O(n\log^*n)$) time!

\item We similarly describe an $O(n\log^*n + n\log\dCross)$-time deterministic algorithm for computing $\DT(P)$ from $G$.  Again, there are examples where $\dCross=1$ but $\Dwrong=\Omega(n)$.  Although the parameter $\dCross$ may not appear to be the most natural at first, it turns out to be the key that allows us to obtain items 2 and 3 above in an elegant, unified way: it is not difficult to show that $\dCross = O((1/\rho)\log n)$ with high probability in the probabilistic model, and we are also able to show that $\dCross \le O(\dVio^2)$ always holds.

\item Although we have focused on the Delaunay triangulation problem,
we can obtain similar results for a few related problems.
For example, consider the problem of computing the Euclidean minimum spanning tree $\EMST(P)$ of a set $P$ of $n$ points in 2D\@.  Given a spanning tree $T$, we can compute
$\EMST(P)$ from $T$ in $O((n + \Dmst\log n)\log^*n)$ expected time, where $\Dmst$ is the number of ``wrong'' edges in $T$ that are not in $\EMST(P)$.
In fact, this result, and the idea of using a spanning tree in the first place, will be instrumental to the eventual design of our $O(n+\Dwrong\log n)$-time randomized algorithm for Delaunay triangulations.

We get similar results for EMST in the probabilistic model or in terms of other parameters, and also for computing other related geometric graphs such as  the
relative neighborhood graph and the Gabriel graph in the plane~\cite{JaromczykT92}.  We also get similar results
for the problem of computing the convex hull of $n$ points in 3D that are in convex position (by a standard lifting map, Delaunay triangulations in 2D maps to convex hulls in 3D)---the convex position assumption is necessary since one could show an $\Omega(n\log n)$ lower bound for checking whether $n/2$ points are all inside the convex hull of another $n/2$ points, even when we are given the correct convex hull.
\end{enumerate}

\subsection{Related work in computational geometry}

We are not aware of any prior work in computational geometry that explicitly addresses the paradigm of algorithms with predictions, except for a recent paper 
by Cabello and Giannopoulos~\cite{CabelloG24} on an on-line geometric search problem (where predictions are used to obtain competitive ratio bounds).
Implicitly related are \emph{adaptive} data structures, and these have been studied in computational geometry; as one example, see Iacono and Langerman's work~\cite{IaconoL03} on ``proximate'' planar point location (where a query can be answered faster whe given a finger near the answer).

Our work is related to, and may be regarded as a generalizaton of, work on the
\emph{verification} of geometric structures: how to check whether a given answer is indeed correct, using a simpler or faster algorithm than an algorithm for computing the answer.
This line of work has received much attention from computational geometers in the 1990s~\cite{DevillersLPT98,McConnellMNS11,MehlhornNSSSSU99}, motivated by algorithm engineering.  For the Delaunay triangulation problem, it is easy to verify whether $G$ is
equal to $\DT(G)$ in linear time, by checking that all edges are locally Delaunay.  Verifying correctness corresponds to the case $\Dwrong=0$, and
our results provides more powerful extensions: we can correct errors (as many as $n/\log n$ wrong edges) in linear time.  Viewing our results in the lens of \emph{correcting errors in the output} provides further motivation, even for readers not necessarily interested in learning-augmented algorithms.

Some recent work such as \cite{EppsteinGS25} studies geometric algorithms in models that tolerate errors in the computation process itself (e.g., noisy or faulty comparisons), but our work seems orthogonal to this research direction.  On the other hand, there has been considerable prior work in computational geometry on models that capture \emph{errors or imprecisions in the input}.  For example, in one result by L\"offler and Snoeyink~\cite{LofflerS10}, an imprecise input point is modelled as a range, specifically, a unit disk (whose center represents its estimated position);
it is shown that $n$ disjoint unit disks in the plane can be preprocessed in $O(n\log n)$ time, so that given $n$ actual points lying in these $n$ disks, their Delaunay triangulation can be computed in $O(n)$ time. Similar results have been obtained for different types of regions and different problems.  Our model provides an alternative way to address imprecise input: in the preprocessing phase, we can just compute the Delaunay triangulation $G$ of the estimated points; when the actual points $P$ are revealed, we can correct $G$ to obtain $\DT(P)$ using our algorithms.\footnote{
One technical issue is that $G$ may have self-crossings when re-embedded into the actual point set $P$, but
at least one of our algorithms works even when $G$ has self-crossings (see Corollary~\ref{cor:DT:EMST:self:cross}).
}
 An advantage of our model is its simplicity: we don't need to be given any geometric range or distribution for each input point; rather, accuracy is encapsulated into a single parameter~$\Dwrong$.

Even without relevance to algorithms with paradigms or coping with errors in the input/output,
Question~1 is appealing purely as a classical computational geometry problem. In the literature, a number of papers have explored conditions under which Delaunay triangulations can be computed in linear or almost linear time.
For example,
Devillers~\cite{Devillers92} showed that given the Euclidean minimum spanning tree of a 2D point set, one can compute the Delaunay triangulation in $O(n\log^*n)$ expected time
(this was later improved to deterministic linear time~\cite{ChinW98});
two Delaunay triangulations $\DT(P_1)$ and $\DT(P_2)$ can be merged into $\DT(P_1\cup P_2)$ in linear time~\cite{Kirkpatrick79};
conversely, a Delaunay triangulation $\DT(P_1\cup P_2)$  can be split into $\DT(P_1)$ and $\DT(P_2)$ in linear  time (by Chazelle et al.~\cite{ChazelleDHMST02}); the
Delaunay triangulation of the vertices of a convex polygon can be computed in linear time (by Aggarwal et al.~\cite{AggarwalGSS89,Chew86});
the Delaunay triangulation for 2D points with integer coordinates can be computed in the same time complexity as sorting in the word RAM model (by Buchin and Mulzer~\cite{BuchinM11}); etc.  
A powerful work by L\"offler and Mulzer~\cite{LofflerM12} unifies many of these results by showing \emph{equivalence} of 2D Delaunay triangulations with quadtree construction and related problems.
Although some of these tools will be useful later, it isn't obvious how they can used to resolve Questions 1 and 2.


\subsection{Our techniques}

Although we view our main contributions to be conceptual (introducing a new research direction in computational geometry), our algorithms are also technically interesting in our opinion.  The solutions turn out to be a lot more challenging than 1D adaptive sorting algorithms or the standard textbook Delaunay triangulation algorithms:
\begin{itemize}
\item
Our deterministic $O(n+\Dwrong\log^4n)$-time algorithm (given in Appendix~\ref{app:det}) makes clever use of dynamic geometric data structures.
\item
Our deterministic $O(n+\Dwrong\log^3n)$-time algorithm  (given in Section~\ref{sec:DT:det}) uses planar graph separators or $r$-divisions in the given graph $G$ (for classical computational geometry problems, Clarkson--Shor-style sampling-based divide-and-conquer~\cite{Clarkson92,ClarksonS89,CheongMR} is more commonly used 
than planar graph separators).
\item
Our randomized $O((n+\Dwrong\log n)\log^*n)$-time algorithm (given in Section~\ref{sec:rand:logstar}) is obtained via an adaptation of the abovementioned randomized $O(n\log^*n)$-time algorithm for computing the Delaunay triangulation from the EMST, by Devillers~\cite{Devillers92} (which in turn was based on the randomized $O(n\log^*n)$-time algorithms by Seidel~\cite{Seidel91} and Clarkson, Tarjan, and Van Wyck~\cite{ClarksonTW89} for the polygon triangulation problem).
The adaptation is more or less straightforward, but our contribution is in realizing that these randomized $O(n\log^*n)$-time algorithms are relevant in the first place.
\item
Our final randomized $O(n+\Dwrong\log n)$-time algorithm (given in Section~\ref{sec:DT:rand}) eliminates the
extra $\log^*n$ factor by incorporating a more complicated divide-and-conquer involving
planar graph separators or $r$-divisions, this time applied to actual Delaunay triangulations, not the input graph $G$. (Curiously, ideas from certain
output-sensitive convex hull algorithms~\cite{EdelsbrunnerS91,ChanSY97} turn out to be useful.)
\item
For our deterministic $O(n\log^*n + n\log\dCross)$-time algorithm (given in Section~\ref{sec:cross}), which as mentioned is the key to our result in the probabilistic model as well as our $\dVio$-sensitive result,
we return to using $r$-divisions in the given graph $G$.  The way the $\log^*n$ factor comes up is interesting---it arises not from the use of known randomized $O(n\log^*n)$ algorithms like Seidel's, but from an original recursion.
Although our $\Dwrong$-sensitive algorithm does not directly yield a $\dCross$-sensitive result, it will turn out to be an important ingredient in the final algorithm.
\end{itemize}

\subsection{Open questions}

Many interesting questions remain from our work (which we see as a positive, as our goal is to inspire further work in this research direction):

\begin{itemize}
\item Could the extra $\log^*n$ factors be eliminated in some of our results?
\item Could the extra $\log\log n$ factor in our result in the probabilistic model be removed?
\item Could we get a time bound in terms of $\Dflip$ better than $O(n+\Dflip\log n)$?  For example, is 
there a Delaunay triangulation algorithm that runs in $O(n+\Dflip)$ time? 
\item Similarly, is there a Delaunay triangulation algorithm that runs in $O(n+\Dcross)$  time?  Or, better still, 
$O(n\log(1+\Dcross/n))$ time?  Or $O(n+\Dvio)$ or $O(n\log(1+\Dvio/n))$ time?
\end{itemize}

More generally, we hope that our work will prompt other models of prediction to be studied for the Delaunay triangulation problem, or similar models for other fundamental problems in computational geometry and beyond.


\section{A deterministic $O(n+\Dwrong\log^3 n)$ algorithm}\label{sec:DT:det}

We begin by presenting a deterministic algorithm for computing the Delaunay triangulation in 
$O(n+\Dwrong\log^3 n)$ time.  The algorithm uses planar graph separators, or more generally, $r$-divisions, on the input triangulation $G$.  (A different deterministic algorithm is given in Appendix~\ref{app:det}, which was actually the first algorithm we discovered; it is slightly slower but does not require separators.)

An \DEF{$r$-division} of a planar graph $G$ with $n$ vertices is a partition of $G$ into $O(n/r)$ regions (edge-induced subgraphs) such that each region has at most $O(r)$ vertices and at most $O(\sqrt{r})$ boundary vertices, i.e., vertices shared by more than one region. An $r$-division can be computed in $O(n)$ time; see, for example, the classical result by Federickson~\cite{Federickson87} or the more recent, improved result by Klein, Mozes, and Sommer~\cite{KleinMS13}. Note that the regions may have holes (i.e., faces that are not faces of $G$). We may also assume that each region is connected.




\newcommand{\good}{\textit{good}}
\newcommand{\bad}{\textit{bad}}

\begin{theorem}\label{thm:det:DT}
Given a triangulation $G$ of a set $P$ of $n$ points in the plane, the Delaunay triangulation $\DT(P)$ can be constructed in $O(n + \Dwrong\log^3 n)$ deterministic time, where $\Dwrong$ is the (unknown) number of edges in $G$ that are not in $\DT(P)$. 
\end{theorem}
\begin{proof}
Our algorithm proceeds in the following steps:

\begin{enumerate}
\item
First we compute, in $O(n)$ time, a $t$-division $\Gamma$ of the dual graph of the input triangulation $G$, for some $t$ to be chosen later. 
Each region $\gamma$ of $\Gamma$ is a connected set of $O(t)$ triangles.  Let $V_\gamma$ denote the set of vertices of all the triangles in
$\gamma$ and $B_\gamma$ denote the set of all the triangles surrounding $\gamma$ (here, it is more convenient to work with boundary triangles instead of boundary edges). Then $|V_\gamma| = O(t)$ and $|B_\gamma| = O(\sqrt{t})$.
Also, let $B$ be the set of all boundary triangles in $\Gamma$ and $V_B$ be the set of the vertices of the edges in $B$. Then $|V_B| = O(n/\sqrt{t})$. 

\item
We compute the Delaunay triangulation $\DT(V_B)$ of $V_B$, in
$O(|V_B|\log|V_B|) \le  O((n/\sqrt{t}) \allowbreak \log n)$ time.

\item
We call a region $\gamma$ of $\Gamma$ \DEF{good} if (i) all boundary triangles in $B_\gamma$ are in $\DT(V_B)$, and (ii) the triangles of $\gamma$ and their immediate neighborhood are locally Delaunay (i.e., for every neighboring pair of triangles $\triangle p_1p_2p_3$ and $\triangle p_1p_2p_4$ in $G$ with at least one inside $\gamma$ (the other may be inside or be a boundary triangle), we have $p_4$ outside the circumcircle through $p_1,p_2,p_3$); we call it \DEF{bad} otherwise. Let
\[ V_{\good} = \left(\bigcup_{\gamma \,\mathrm{is \, good}}V_\gamma\right) \setminus V_B\ \ \mbox{and}\ \ V_{\bad}= \left(\bigcup_{\gamma \,\mathrm{is \, bad}}V_\gamma\right)\setminus V_B. 
\]
Note that condition (i) can be tested using standard data structures, while condition (ii) can be tested in linear total time.


\item
We compute $\DT(V_B\cup V_{\good})$.
This can be done by modifying $\DT(V_B)$.  Namely, for each good region
$\gamma$, we replace what is inside $B_\gamma$ (by definition,
$B_\gamma$ belongs to $\DT(V_B)$), with the triangles of~$\gamma$.  Then
the resulting triangulation is locally Delaunay, and thus must coincide
with $\DT(V_B\cup V_{\good})$ (since the local Delaunay property holding for the entire triangulation implies
that it is the Delaunay triangulation~\cite{BergCKO08}). This step takes linear time.

\item
We compute $\DT(V_{\bad})$.
We will do this step from scratch, in $O(|V_{\bad}|\log|V_{\bad}|)$ time.
For each bad region $\gamma$, there is at least one non-Delaunay edge of $G$
inside $\gamma$ or on its boundary triangles.  So, the number of bad regions is
at most $O(\Dwrong)$, implying that $|V_{\bad}|\le O(\Dwrong t)$. So, this step takes
$O(\Dwrong t\log n)$ time.

\item
Finally, we merge $\DT(V_B\cup V_{\good})$ and $DT(V_{\bad})$, in linear
time by standard algorithms for merging Delaunay triangulations
(Chazelle~\cite{Chazelle92}, and later Chan~\cite{Chan16}, showed more generally that two convex hulls in $\RR^3$
can be merged in linear time, but an earlier result by Kirkpatrick~\cite{Kirkpatrick79} suffices for
the case of Delaunay triangulations in $\RR^2$).
\end{enumerate}

The total running time is $O(n + (n/\sqrt{t} + \Dwrong t)\log n)$.
Choosing $t=\log^2n$ gives the desired $O(n + \Dwrong \log^3 n)$ bound.
\end{proof}

\section{A randomized $O((n+\Dwrong\log n)\log^* n)$ algorithm}\label{sec:rand:logstar}


In this section, we present a completely different, and faster, randomized algorithm for computing the Delaunay triangulation in terms of the parameter $\Dwrong$, which is just a $\log^*n$ factor from optimal.
The algorithm is an adaptation of an algorithm by Devillers \cite{Devillers92}, which
uses the Euclidean minimum spanning tree to construct the Delaunay triangulation
in almost linear ($O(n\log^*n)$) expected time.
Slightly more generally, the Delaunay triangulation can be constructed
from any given spanning tree of $\DT(P)$ of maximum degree $O(1)$ in $O(n\log^*n)$ expected time.
We generalize his result further to show that the same can be achieved even if the given spanning tree has some wrong edges,
and even if the maximum degree requirement is dropped.  


Devillers' algorithm is similar to the more well-known
$O(n\log^*n)$ randomized algorithms for polygon triangulation by Seidel~\cite{Seidel91}
and Clarkson, Tarjan, and Van Wyk~\cite{ClarksonTW89}, which are
variants of standard $O(n\log n)$ randomized incremental or random sampling algorithms~\cite{Clarkson92,ClarksonS89,CheongMR}, but with the point location steps sped up.
We will describe our algorithm based on random sampling.


It will be more convenient to work in 3D, where 2D Delaunay triangulations and Voronoi diagrams are mapped to
3D convex hulls and lower envelopes of planes.  We begin with some standard terminologies:

For a set of planes $\Pi$ in $\RR^3$, let $\LE(\Pi)$ denote the \DEF{lower 
envelope of $\Pi$}, that is, the set of points lying on some plane of $\Pi$
but with no plane of $\Pi$ below it.
The \DEF{vertical decomposition of $\LE(\Pi)$}, denoted $\VDLE(\Pi)$,
is a partition of the space below $\LE(\Pi)$ into vertical, truncated prisms 
that are unbounded from below and are bounded from above by a triangular 
face contained in $\LE(\Pi)$. The triangles bounding the prisms have
vertices on $\LE(\Pi)$ and are defined by some canonical rule, 
such as triangulating each face of $\LE(\Pi)$ with diagonals from the point
in the face with the smallest $x$-coordinate. (Since the faces are convex,
such a triangulation is possible.) We also have to ``triangulate''
the infinite faces of $\LE(\Pi)$, that is, to split them
into regions defined by three edges.  We will refer to these prisms as \DEF{cells}.

For each point $p = (\alpha, \beta)\in \RR^2$, 
let $p^*$ denote the plane in $\RR^3$
defined by $z= - 2\alpha x - 2\beta y + \alpha^2+\beta^2$.
The plane $p^*$ is tangent to the paraboloid $z=-x^2-y^2$ 
at the point $p$. For any set $P$ of points, define 
$P^*=\{ p^*\mid p\in P\}$.

For a set $P$ of points in the plane, 
the $xy$-projection of the lower envelope of $P^*$
is the Voronoi diagram of $P$ (dual to the Delaunay triangulation $\DT(P)$). The vertices and edges of $\LE(P^*)$
correspond to the vertices and edges of the Voronoi diagram of $P$
in the plane. Since we assume general position,
each vertex of $\LE(P^*)$ is contained in three planes of $P^*$.
The $xy$-projection of a cell
in $\VDLE(P^*)$ is a triangle
contained in a Voronoi cell and defined by three Voronoi vertices.

For a vertex $v=p_1^*\cap p_2^*\cap p_3^*\in\RR^3$, the \DEF{conflict list of $v$}, denoted $P^*_v$, refers to the list of 
all planes in $P^*$ strictly below $v$.  
Observe that $p^*$ is in the conflict list $P^*_v$ iff $p$ is strictly inside
the circumcircle through $p_1,p_2,p_3$.

\begin{lemma}\label{lem:conflict}
	Let $p_1,p_2,p_3$ be the vertices of the triangle $\triangle$, 
	and let $p,q$ be points such that $pq$ intersects 
	the interior of the triangle $\triangle$.
	If $pq$ is an edge of $DT(\{ p_1,p_2,p_3, p,q\})$,  
	then $p^*$ or $q^*$ must be in the conflict list of $p_1^*\cap p_2^*\cap p_3^*$.
\end{lemma}
\begin{proof}
	\begin{figure}\centering
		\includegraphics[page=1]{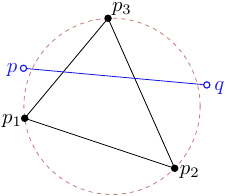}
		\caption{Proof of Lemma~\ref{lem:conflict}.}
		\label{fig:conflict}
	\end{figure}
	See Figure~\ref{fig:conflict}.
	If $p$ and $q$ are not inside the circumcircle of $\triangle$, then
	the edges of $\triangle$ are edges of $\DT(\{ p_1,p_2,p_3, p,q\})$.
	Since the endpoints of $pq$ are not inside $\triangle$,
	$pq$ crosses some edge of $\triangle$. This would contradict the planarity of $\DT(\{ p_1,p_2,p_3, p,q\})$.
\end{proof}

For a cell $\tau\subset \RR^3$, the \DEF{conflict list of $\tau$}, denoted $P^*_\tau$, refers to
the list of all planes in $P^*$ intersecting $\tau$.
Observe that the conflict list of $\tau$ is just the union of the conflict lists
of the three vertices of $\tau$.


Before describing our main algorithm, we first present the following lemma, which uses a given spanning tree $T$ to speed up
conflict list computations.  
To deal with the scenario when the maximum degree of $T$ is not bounded, we will work with weighted samples, or equivalently, samples of multisets.


\newcommand{\PPP}{\widehat{P}}
\newcommand{\RRR}{\widehat{R}}

\begin{lemma}\label{lem:all_conflicts}
	Let $P$ be a set of $n$ points, and $T$ be a spanning tree of $P$, where the edges around each vertex have been
	angularly sorted.  
	Let $\PPP$ be the multiset where each point $p\in P$ has multiplicity $\DEG_T(p)$.
	Let $\RRR$ be a random (multi)subset of $\PPP$ of a given size $r$, and let $R$ be the set $\RRR$ with duplicates removed.  
	Suppose $\LE(R^*)$ is given.
	Then we can compute the conflict lists of all the cells in $\VDLE(R^*)$ in
	$O(n+\Dtree\log n)$ expected time, where $\Dtree$ is the number of edges of $T$ that are not in $\DT(P)$.
\end{lemma}
\begin{proof}
%
	We preprocess $DT(R)$ for point location~\cite{Kirkpatrick83}. 
	The preprocessing takes $O(r)$ time and 
	a point location query can be answered in $O(\log r)\le O(\log n)$ time.
	
    For each point $p\in P$, let $\gamma_p$ be the number
	of vertices of $\VDLE(R^*)$ whose conflict lists contain~$p^*$.
	
	We first describe how to compute, for each point $p\in P$,
	the triangle $\triangle_p\in \DT(R)$ containing $p$.  (Technically, we should add 3 vertices of a sufficiently large bounding triangle to both $R$ and $P$, to ensure that $\triangle_p$ always exists, and to deal with vertices of $\VDLE(R^*)$ ``at infinity''.)
	
	To this end, we traverse the edges of the spanning tree $T$.
	For each edge $pq$ of $T$, knowing $\triangle_p$ we can get to $\triangle_q$ 
	by performing a linear walk in $\DT(R)$. However, if we cross more than
	$\log n$ triangles of $\DT(R)$, we switch to performing a point
	location query for $q$. This takes $O(\min \{ c_{pq},\log n\})$ time, 
	where $c_{pq}$ is the number of triangles of $\DT(R)$ crossed by $pq$. 
	Iterating over all edges in $T$, we get $\triangle_p$
	for all $p\in P$.  The total time is 
	\begin{equation}\label{eq:sum1}
		O\left( \sum_{pq\in T}\min \{ c_{pq},\log n\}\right).
	\end{equation}
	If $pq$ is an edge of $\DT(P)$, then \cref{lem:conflict} implies
	that $c_{pq}$ is bounded by $\gamma_p+\gamma_q$. Therefore 
	\begin{align*}
		\sum_{pq\in T\cap \DT(P)} c_{pq}\ \le\ 
			\sum_{p\in P} ( \gamma_p\cdot \deg_T(p)),
	\end{align*} 
	which is proportional to the total conflict list size, with multiplicities included.
	By Clarkson and Shor's analysis~\cite{ClarksonS89}, the expected total conflict list size, 
	with multiplicities included, is $O(r\cdot |\PPP|/r) = O(n)$, where we are using the fact 
	that $|\PPP|=O(\sum_p \deg_T(p))=O(n)$. 
	
	On the other hand,
	\begin{align*}
		\sum_{pq\in T\setminus \DT(P)} \min \{ c_{pq},\log n\} \ \le\ O(\Dtree\log n).
	\end{align*}
	Plugging these in \cref{eq:sum1}, we get the claimed expected time bound.

	One case not considered above is when $p\in R$.  To begin the walk from $p$ to $q$ for a given edge $pq$ in $T$, we need to find the triangle $\triangle_{pq}\in\DT(R)$ incident to $p$ that is crossed by the ray $\RAY{pq}$.  To this end, during preprocessing, we can merge the sorted list of the edges of $T$ around $p$ with the sorted list of edges of $\DT(R)$ around $p$, for every $p\in R$; the total cost of merging is linear.  Afterwards, each triangle $\triangle_{pq}$ can be identified in constant time.

	Now, $\triangle_p$ maps to a vertex of $\LE(P^*)$ whose conflict list contains $p^*$.
	By a graph search, we can generate all vertices of $\LE(P^*)$ whose conflict lists contain $p^*$, for each $p\in P$.  By inversion, we obtain the conflict lists of all vertices of $\LE(P^*)$, from which we 
	get the conflict lists of all cells in $\VDLE(P^*)$.  All this takes
	additional time linear in the total conflict list size, which is $O(n)$ in expectation as mentioned by Clarkson and Shor's bound.
\end{proof}

\begin{theorem}\label{thm:EMST}
Given a spanning tree $T$ of a set $P$ of $n$ points in the plane, 
the Delaunay triangulation $\DT(P)$ can be constructed in 
$O((n + \Dtree\log n)\log^*n)$ expected time, where $\Dtree$ is the number of 
edges of $T$ that are not in $\DT(P)$.
\end{theorem}
\begin{proof}
    Let $\PPP$ be the multiset where each point $p\in P$ has multiplicity $\DEG_T(p)$; note that $|\PPP|=O(n)$.
    Define a sequence of random samples $\RRR_1\subset \RRR_2\subset \cdots\subset \RRR_\ell\subset \PPP$,
    where $\RRR_i$ has size $|\PPP|/s_i$ for a sequence $s_1,s_2,\ldots,s_\ell$ to be chosen later.
    Let $R_i$ be the set $\RRR_i$ with duplicates removed.
    
    Given $\LE(R_i^*)$, we describe how to compute $\LE(R_{i+1}^*)$.
    First, we compute the conflict lists of all cells in $\VDLE(R_i^*)$ by \cref{lem:all_conflicts}
    in $O(n+\Dtree\log n)$ expected time.
    For each cell $\tau\in \VDLE(R_i^*)$, we compute
    $\LE((R_{i+1}^*)_\tau) \cap\tau$
    by a standard lower envelope algorithm in $O(|(R_{i+1}^*)_\tau|\log |(R_{i+1}^*)_\tau|)$ time.
    Finally, we glue all these envelopes to obtain $\LE(R_{i+1}^*)$.
    The running time is 
    \[ O\left(\sum_{\tau\in\VDLE(R_i^*)}|(R_{i+1}^*)_\tau|\log |(R_{i+1}^*)_\tau|\right),\]
    which by Clarkson and Shor's bound~\cite{ClarksonS89} has expected value at most
    \[ 
		O(|\RRR_i| \cdot (|\RRR_{i+1}|/|\RRR_i|) \log (|\RRR_{i+1}|/|\RRR_i|)) \ =\ O((n/ s_{i+1}) \log (s_i/s_{i+1})).
		\]
    The total expected time is upper-bounded by
    $O(\ell (n + \Dtree\log n) + \sum_{i=1}^{\ell-1}  (n/s_{i+1})\log s_{i})$.
    Setting $s_1=n$, $s_{i+1}=\lfloor\log s_i\rfloor$, \ldots, $s_\ell=1$, with $\ell=O(\log^*n)$ yields
    the theorem.
\end{proof}

\begin{corollary}\label{cor:DT:EMST}
Given a triangulation $G$ of a set $P$ of $n$ points in the plane, the Delaunay triangulation $\DT(P)$ can be constructed in $O((n + \Dwrong\log n)\log^*n)$ expected time, where $\Dwrong$ is the (unknown) number of edges in $G$ that are not in $\DT(P)$. 
\end{corollary}
\begin{proof}
Let $T$ be any spanning tree of $G$.
We just apply Theorem~\ref{thm:EMST}, noting that $\Dtree\le\Dwrong$ (since $T\subseteq G$, and edges in $G$ around each vertex  are already sorted).
%
\end{proof}

(Chazelle and Mulzer~\cite{ChazelleM11} obtained a similar expected time bound $O(n\log^*n + k\log n)$ for a different problem, on computing the Delaunay triangulation of a subset of $n$ edges with $k$ connected components, given the Delaunay triangulation of a possibly much larger point set.)

Since Theorem~\ref{thm:EMST} holds even if $T$ has self-crossings, our result holds even if the given $G$ is 
a combinatorial triangulation that may have self-crossings.

\begin{corollary}\label{cor:DT:EMST:self:cross}
Given a set $P$ of $n$ points in the plane and a combinatorial triangulation $G$ with vertex set $P$ that may have self-crossings, the Delaunay triangulation $\DT(P)$ can be constructed in $O((n + \Dwrong\log n)\log^*n)$ expected time, where $\Dwrong$ is the (unknown) number of edges in $G$ that are not in $\DT(P)$. 
\end{corollary}
\begin{proof}
Let $d_p$ be the number of triangular faces in $G$ incident to $p$ that are not triangles in $\DT(P)$.
The neighbors of each vertex $p$ in $G$ are almost sorted correctly except possibly for $O(d_p)$ displaced elements;
they can be sorted in $O(\DEG_G(p) + d_p\log d_p)$ time by known adaptive sorting algorithms~\cite{Estivill-CastroW92}.
The total sorting time is thus $O(n + \sum_p d_p\log n) = O(n+\Dwrong\log n)$.

Let $T$ be any spanning tree of $G$.  We can now apply Theorem~\ref{thm:EMST} to $T$ as before.
\end{proof}

Incidentally, we also obtain an algorithm for EMST sensitive to the number of wrong edges:

\begin{corollary}\label{cor:EMST}
Given a spanning tree $T$ of a set $P$ of $n$ points in the plane, 
the Euclidean minimum spanning tree $\EMST(P)$ can be constructed in 
$O((n + \Dmst\log n)\log^*n)$ expected time, where $\Dmst$ is the (unknown) number of 
edges of $T$ that are not in $\EMST(P)$.	
\end{corollary}
\begin{proof}
	Since the maximum degree of $\EMST(P)$ is at most $6$, 
	every vertex $p$ of degree at least $7$ in $T$ is adjacent 
	to at least $\deg_T(p) - 6$ edges not in $\EMST(P)$. This implies that $\sum_p \max\{\deg_T(p)-6,0\} \le O(\Dmst)$.
	We can sort the edges of $T$ around each vertex, in total time $O(\sum_p \deg_T(p)\log\deg_T(p)) \le O(n + \sum_p \max\{\deg_T(p)-6,0\}\log n) = O(n+\Dmst\log n)$. We can now apply \cref{thm:EMST}, noting that $\Dtree\le \Dmst$ 
	(since $\EMST(P)\subseteq\DT(P)$).


	Having computed $\DT(P)$, we can then obtain $\EMST(P)$ by computing a minimum spanning
tree of $\DT(P)$. For this, we can use the linear-time algorithms
for planar graphs by Matsui~\cite{Matsui95} or by Cheriton and Tarjan~\cite{CheritonT76}
(the general, linear-time randomized algorithm of Karger, Klein 
and Tarjan~\cite{KargerKT95} is not necessary).
\end{proof}

We can also obtain a similar result for the \DEF{relative neighborhood graph}~\cite{JaromczykT92}:
$pq$ forms an edge iff the intersection of the two disks centered at $p$ and $q$ of radius
$\|p-q\|$ does not have any points of $P$ inside. The relative neighborhood graph
contains the Euclidean minimum spanning tree and is thus connected.

\begin{corollary}
Given a set $P$ of $n$ points in the plane and a connected PSLG\footnote{PSLG stands for a plane straight-line graph; see Appendix~\ref{app:det} for a definition.
} $G$ with vertex set $P$,  
the relative neighborhood graph $\RNG(P)$ can be constructed in 
$O((n + \Drng\log n)\log^*n)$ expected time, where $\Drng$ is the (unknown) number of 
edges of $G$ that are not in $\RNG(P)$.	
\end{corollary}
\begin{proof}
Let $T$ be any spanning tree of $G$.
We just apply Theorem~\ref{thm:EMST}, noting that $\Dtree\le\Drng$ (since $T\subseteq G$ and $\RNG(G)\subseteq \DT(G)$, and edges in $G$ around each vertex  are already sorted).
Having computed $\DT(P)$, we can then obtain $\RNG(G)$ in linear time~\cite{Lingas94}.
\end{proof}

A similar result holds for the Gabriel graph or, more generally, the $\beta$-skeleton for any $\beta\in[1,2]$~\cite{JaromczykT92}.



\section{An optimal randomized $O(n+\Dwrong\log n)$ algorithm}\label{sec:DT:rand}


In this section, we present our optimal randomized algorithm  for computing the Delaunay triangulation in terms of the parameter $\Dwrong$, eliminating the extra $\log^*n$ factor from the algorithm in the previous section.
This is interesting since for related problems such as polygon triangulation, removing the extra $\log^*n$ factor from Clarkson, Tarjan, and Van Wyk's or Seidel's
algorithm~\cite{ClarksonTW89,Seidel91} is difficult to do (cf.\ Chazelle's linear-time
polygon triangulation algorithm~\cite{Chazelle91}).

Our new approach is to stop the iterative algorithm in the proof of \cref{thm:EMST} earlier after a constant number of rounds,  when the conflict list sizes gets sufficiently small (say, $O(\log\log n)$), and then switch to a different strategy. For this part, we use planar graph separators ($t$-divisions) on actual Delaunay triangulations (of random subsets).

\begin{theorem}\label{thm:rand:opt}
	Given a triangulation $G$ of a set $P$ of $n$ points in the plane, 
	the Delaunay triangulation $\DT(P)$ can be constructed in $O(n + \Dwrong\log n)$ expected time, 
	where $\Dwrong$ is the (unknown) number of edges in $G$ that are not in $\DT(P)$.
\end{theorem}
\begin{proof}
Note that at most $O(\Dwrong)$ of the triangles of $G$ are not in $\DT(P)$, or equivalently, at most $O(\Dwrong)$ of the triangles of $\DT(P)$ are not in $G$  (since $G$ and $\DT(P)$ have the same number of triangles).
 
Call a triangle of $\DT(P)$ \DEF{good} if it is in $G$, and \DEF{bad} otherwise.

Let $T$ be any spanning tree of $G$.
Since $T$ has at most $\Dwrong$ edges not in $\DT(P)$, we can apply the algorithm in 
\cref{thm:EMST}.
However, we will stop the algorithm after $\ell$ iterations with a nonconstant choice of $s_\ell$ instead of $s_\ell=1$.
Once we have computed $\LE(R_\ell^*)$, we will switch to a different strategy
to finish the computation of $\LE(P^*)$, as described below.

First, we construct the conflict lists of $\VDLE(R_\ell^*)$ by \cref{lem:all_conflicts}
in $O(n + \Dwrong\log n)$ expected time.

Call a cell \DEF{exceptional} if it has conflict list size more than $b s_\ell$ for some parameter $b$ to be set later.
We directly construct the lower envelope $\LE(P_\tau^*)\cap\tau$ in $O(|P_\tau^*|\log|P_\tau^*|)$ time 
for all exceptional cells $\tau\in\VDLE(R_\ell^*)$.
By Chazelle and Friedman's ``Exponential Decay Lemma''~\cite{ChazelleF90},
the expected number of cells with conflict list size more than $is_\ell$ is $O(2^{-\Omega(i)}\cdot n/s_\ell)$,
and thus the expected total cost of this step is at most
$O(\sum_{i\ge b} 2^{-\Omega(i)}\cdot (n/s_\ell)\cdot (is_\ell)\log(is_\ell))
= O(2^{-\Omega(b)}\cdot n \log s_\ell)$.

We compute a $t$-division of the dual planar graph of the $xy$-projection of $\VDLE(R_\ell^*)$ in $O(n)$ time
for some parameter $t$.
This yields $O(n/(s_\ell t))$ regions, which we will call \DEF{supercells},
each of which consists of $O(t)$ cells in $\VDLE(R_\ell^*)$ and has $O(\sqrt{t})$ boundary walls.

We compute $\LE(P^*_\beta)\cap\beta$ for each boundary wall $\beta$.
This costs $O(n/(s_\ell t) \cdot \sqrt{t}\cdot (bs_\ell)\log(bs_\ell))$ time
(since exceptional cells have already been taken care of).

The vertices of $\LE(P^*)$ at its boundary walls map to chains of 
edges in $\DT(P)$ when returning to $\RR^2$.  (Some past output-sensitive convex hull algorithms
also used these kinds of chains for divide-and-conquer~\cite{EdelsbrunnerS91,ChanSY97,AmatoGR94}, but we will use them in a different way.)
These edges divide $\DT(P)$ into regions, where each region is associated to one supercell
(although one supercell could be associated with multiple regions).
We  can remove ``spikes'' (i.e., an edge where both sides bound the same region). 
For each such region, we check whether all its boundary edges appear in $G$,
and if so, check whether all triangles in $G$ inside the region are locally Delaunay
(i.e., for every neighboring pair of triangles $\triangle p_1p_2p_3$ and $\triangle p_1p_2p_4$ in $G$ inside the region,
we have $p_4$ outside the circumcenter through $p_1,p_2,p_3$).
If the checks are true, then we know that all triangles in $G$ inside the region are correct
Delaunay triangles,\footnote{
If a triangulation $G$ satisfies the local Delaunay property inside a region $\gamma$, and all the boundary edges of $\gamma$ are correct Delaunay edges, then the part of $G$ inside $\gamma$ are correct Delaunay triangles.
This fact follows immediately, for example, from known properties on constrained Delaunay triangulations~\cite{LeeL86}, with $\partial\gamma$ as constrained edges.
} and will mark the region as such.  All these checks take linear total time.

For each supercell $\gamma$, if its corresponding regions are marked correct, then we have
all the vertices in $\LE(P^*)\cap\gamma$ for free, and we say that the supercell $\gamma$ is \DEF{good}.
Otherwise, the supercell $\gamma$ is \DEF{bad}, and we just
compute the vertices in $\LE(P^*)\cap\gamma$
by computing $\LE(P_\tau^*)\cap\tau$ for each cell $\tau$ of $\gamma$,
in time $O(t\cdot (bs_\ell)\log(bs_\ell))$ (again, exceptional cells have already
been taken care of).

Observe that if the Delaunay triangles corresponding to all vertices in $\LE(P^*)\cap\gamma$ are
all good, then the supercell $\gamma$ must be good.
Thus, each bad supercell contains at least one vertex in $\LE(P^*)$ whose Delaunay triangle is bad.
So, the number of bad supercells is at most $O(\Dwrong)$.  

We conclude that from $\LE(P_\ell^*)$, we can compute $\LE(P^*)$ in expected time
\[  O(  n + \Dwrong\log n + n/(s_\ell t) \cdot \sqrt{t}\cdot (bs_\ell)\log(bs_\ell) + \Dwrong\cdot t\cdot (bs_\ell)\log(bs_\ell) + 2^{-\Omega(b)} \cdot n\log s_\ell ).
\]
For example, we can choose $b=s_\ell$ and $t=s_\ell^4$ and obtain the time bound
$O(n + \Dwrong\log n + \Dwrong s_\ell^{O(1)})$.

Recall that the algorithm from the proof of Theorem~\ref{thm:EMST} has expected time bound
    $O(\ell (n + \Dwrong \log n) + \sum_{i=1}^{\ell-1}  (n/s_{i+1})\log s_{i})$.
    Setting $s_1=n$, $s_2=\log n$, $s_3=\log\log n$ with $\ell=3$ yields
    overall time bound $O(n + \Dwrong \log n + \Dwrong(\log\log n)^{O(1)})=O(n+\Dwrong\log n)$.
\end{proof}

The same result holds also when the given graph $G$ is not a triangulation but just a PSLG (possibly disconnected),
with missing edges.

\begin{corollary}\label{cor:rand:opt}
	Given a set $P$ of $n$ points in the plane and a PSLG $G$ with vertex set $P$, 
	the Delaunay triangulation $\DT(P)$ can be constructed in $O(n + \Dwrong\log n)$ expected time, 
	where $\Dwrong$ is redefined as the (unknown) number of edges in $\DT(G)$ that are not in $G$.
\end{corollary}
\begin{proof}
The number of missing edges is at most $\Dwrong$, and one can easily first triangulate the non-triangle faces of $G$ in $O(\Dwrong\log \Dwrong)$ time and then run our algorithm in Theorem~\ref{thm:rand:opt}. 
\end{proof}





\section{An $O(n\log^*n + n\log\dCross)$ algorithm}\label{sec:cross}

In this section, we present an algorithm for computing the Delaunay triangulation which is
sensitive to the parameter $\dCross$.  The algorithm uses planar graph separators or $r$-divisions in
the input graph $G$, like in the deterministic algorithm in \cref{sec:DT:det}, but the divide-and-conquer is
different.  The difficulty is that even when $\dCross$ is small, the number $\Dwrong$ of wrong edges
could be large. Our idea is to recursively compute the correct Delaunay triangulation of each region, and
combine them to obtain a new triangulation where the number of wrong edges is smaller ($O(n/\polylog n)$);
afterwards, we can apply our earlier $\Dwrong$-sensitive algorithm to fix this triangulation.  The depth of the recursion
is $O(\log^*n)$ (but this is not related to the $\log^*n$ factor in the randomized algorithm in \cref{sec:rand:logstar}).

We begin with known subroutines on ray shooting:

\begin{lemma}\label{lem:ray:shoot}
Given $h$ disjoint simple polygons with a total of $n$ vertices in the plane, we can build a data structure in $O(n + h^2\log^{O(1)} h)$ time so that a ray shooting query can be answered in $O(\log n)$ time.

More generally, given a parameter $\delta\in[0,1]$, we can build a data structure in $O(n + h^{2-\delta}\log^{O(1)} h)$ time so that a ray shooting query can be answered in $O(h^\delta\log n)$ time.
\end{lemma}
\begin{proof}
The first part is due to Chen and Wang~\cite{ChenW15}.  The second part follows by dividing the input into $h^\delta$ subcollection of $h^{1-\delta}$ polygons each, and applying the data structure to each subcollection: the preprocessing time becomes $O(n + h^\delta\cdot (h^{1-\delta})^2\log^{O(1)} h)$ and the query time becomes $O(h^\delta\log n)$ (since
ray shooting is a ``decomposable'' problem---knowing the answers to a query for each of the $h^\delta$ subcollections, we would know the overall answer).
\end{proof}

\begin{lemma}\label{lem:intersect}
Given a connected PSLG $G$ with $n$ vertices
and a set $Q$ of line segments, we can compute the subset $A$ of all edges of $G$ that
 cross $Q$, in $O(n+(|Q|+|A|)^{2-\delta}\log^{O(1)}n)$ time for a sufficiently small constant $\delta>0$.
\end{lemma}
\begin{proof}
Let $r\in [1,|Q|]$ be a parameter to be chosen later.
Let $R$ be a \DEF{$(1/r)$-net} of $Q$ of size $O(r\log r)$, i.e.,
a subset $R\subset Q$ such that every line segment crossing at least $|Q|/r$ segments of $Q$ must cross at least one segment of $R$.  By known results on $(1/r)$-nets~\cite{Matousek93}, this takes
$O(|Q|r^{O(1)})$ deterministic time (if randomization is allowed, a random sample already satisfies the property with good probability).

We preprocess $G$ in $O(n)$ time so that a ray shooting query can be answered in $O(\log n)$ time by Lemma~\ref{lem:ray:shoot} (since $G$ can be turned into a tree by introducing cuts, and a tree can be viewed as a simple polygon by traversing each edge twice).
For each segment $q\in R$, we find all edges of $G$ crossing $q$ in time $O(\log n)$ times the number of crossings, by repeated ray shooting.  The total time of this step is at most $O(n + (r\log r)\cdot |A|\cdot \log n)$.

Let $A_R$ denote the subset of all edges of $G$ that cross $R$.
We next preprocess $G-A_R$ in $O(n + |A|^{2-\delta}\log^{O(1)} n)$ time so that a ray shooting query can be answered in $O(|A|^\delta\log n)$ time by Lemma~\ref{lem:ray:shoot} (since $G-A_R$ has at most $O(|A_R|)\le O(|A|)$ components and can be viewed as $O(|A|)$ simple polygons).
For each segment $q\in Q$, we find all edges of $G-A_R$ crossing $q$ in time
$O(|A|^\delta\log n)$ times the number of crossings, by repeated ray shooting. 
Each edge in $G-A_R$ can cross at most $|Q|/r$ segments of $Q$ because $R$ is a $(1/r)$-net; thus, the total number of crossings between $G-A_R$ and $Q$ is at most 
$O(|A|\cdot |Q|/r)$.  The total time of this step is at most
$O(n + |A|^{2-\delta}\log^{O(1)} n + |A|\cdot |Q|/r\cdot |A|^\delta\log n)$.
The lemma follows by setting $r=|Q|^{2\delta}$ and choosing a sufficiently small $\delta>0$ (in our application, it is 
not important to optimize the value of $\delta$).
\end{proof}

%
%

\begin{lemma}\label{lem:planar}
Every (not necessarily connected) planar graph with $n$ vertices has at most $3n-6$ edges and at least $3n-6-O(z)$ edges, where
$z$ is the number of edges appearing in non-triangular faces plus the number of isolated vertices.
\end{lemma}
\begin{proof}
By adding $O(z)$ edges, we can turn the graph into a fully triangulated (i.e., maximal planar) graph, which has exactly $3n-6$ edges.
\end{proof}

\begin{theorem}\label{thm:cross}
Given a triangulation $G$ of a set $P$ of $n$ points in the plane,
the Delaunay triangulation $\DT(P)$ can be constructed
in $O(n\log^* n + n\log \dCross)$ deterministic time,
where $\dCross$ is the (unknown) maximum number of edges of $\DT(P)$ that an edge of $G$ may 
cross.
\end{theorem}
\begin{proof}
Consider a region $\gamma_0$, which is the union of a connected set of at most $\mu$ triangles
of $G$, with at most $c\sqrt{\mu}$ boundary edges, for a sufficiently large constant $c$.  (Initially, $\gamma_0$ is the entire convex hull of $P$.)  
Let $V_{\gamma_0}$ denote the set of vertices of all triangles in $\gamma_0$.
We describe a recursive algorithm to compute $\DT(V_{\gamma_0})$:

\begin{enumerate}
\item
First we compute, in $O(\mu)$ time, a $t$-division $\Gamma$ of the dual graph of the part of $G$ inside $\gamma_0$, for some $t$ to be chosen later. 
Each region $\gamma$ of $\Gamma$ is the union of a connected set of at most $t$ triangles, with 
at most $c\sqrt{t}$ boundary edges.  Let $V_\gamma$ denote the set of vertices of all  triangles in $\gamma$.

\item
For each region $\gamma$, we compute $\DT(V_\gamma)$ recursively.

\item
For each region $\gamma$, let $X_\gamma$ be the set of all edges in $\DT(V_\gamma)$ that cross the boundary $\partial\gamma$, and let $Y_\gamma$ be the set of all edges in $\DT(V_\gamma)$ that are contained in $\gamma$.

We bound the number of ``wrong'' edges in $Y_\gamma$ via an indirect argument, by considering $\DT(P)$:
We know that the number of edges and triangles of $\DT(P)$ crossing $\partial\gamma$ is
at most $O(\dCross\sqrt{t})$.  
The number of edges of $\DT(P)$ contained in $\gamma$ must then be at least $3|V_\gamma|-O(\dCross\sqrt{t})$ by Lemma~\ref{lem:planar}.
Since all edges of $\DT(P)$ contained in $\gamma$ are in both $Y_\gamma$ and in $\DT(V_{\gamma_0})$, 
there are at least  $3|V_\gamma|-O(\dCross\sqrt{t})$ edges of $Y_\gamma$ that are in $\DT(V_{\gamma_0})$.  This implies that there are at most $3|V_\gamma| - (3|V_\gamma|-O(\dCross\sqrt{t})) = O(\dCross\sqrt{t})$ ``wrong'' edges of $Y_\gamma$ that are not in $\DT(V_{\gamma_0})$.

In particular, we also have $|Y_\gamma| \ge 3|V_\gamma|-O(\dCross\sqrt{t})$, and $|X_\gamma| \le 3|V_\gamma|-|Y_\gamma|=O(\dCross\sqrt{t})$.

We can compute $X_\gamma$ by applying Lemma~\ref{lem:intersect} to the triangulation $\DT(V_\gamma)$ and the line segments in $\partial\gamma$, in
$O(t + (\dCross\sqrt{t})^{2-\delta}\log^{O(1)}t)$ time.
We can then generate $Y_\gamma$ in $O(t)$ additional time by a graph traversal.  The total running time for this step is $O(\mu/t \cdot (t + (\dCross\sqrt{t})^{2-\delta}\log^{O(1)}t)) = O(\mu + (\dCross^{2-\delta}\mu/t^{\delta/2})\log^{O(1)}t)$.
 
\item
For each region $\gamma$, we extend $Y_\gamma$ to obtain a complete triangulation $G'_\gamma$ of $\gamma$
using all the vertices in $V_\gamma$ (but no Steiner points), by triangulating the non-triangular faces in $Y_\gamma\cup\partial\gamma$ inside $\gamma$.
Since $\partial\gamma$ has size $O(\sqrt{t})$ and $|Y_\gamma| \ge 3|V_\gamma|-O(\dCross\sqrt{t})$, the non-triangular
faces have complexity $O(\dCross\sqrt{t})$ and  can be triangulated in $O(\dCross\sqrt{t}\log(\dCross\sqrt{t}))$ time.  This cost is absorbed by the cost of other steps.

\item
We combine the triangulations $G'_\gamma$ for all $\gamma$ in $\Gamma$ to get a complete triangulation $G'$
of $\DT(V_{\gamma_0})$, by triangulating the non-triangular faces in $Y_{\gamma_0}\cup\partial\CH(V_{\gamma_0})$.
These non-triangular faces
have total complexity $O(\sqrt{\mu})$, and so can be triangulated in $O(\sqrt{\mu}\log\mu)$ time.

\item
Finally, apply our previous algorithm from Theorem~\ref{thm:det:DT} to obtain $\DT(V_{\gamma_0})$ from $G'$.

The total number of ``wrong'' edges of $G'$ that are not in $\DT(V_{\gamma_0})$ is
at most $O(\mu/t \cdot \dCross\sqrt{t}  + \sqrt{\mu}) = O(\dCross\mu/\sqrt{t} + \sqrt{\mu})$.  So, this step takes
$O(\mu + (\dCross\mu/\sqrt{t} + \sqrt{\mu}) \log^3\mu)$ deterministic time.
\end{enumerate}

The total time satisfies the following recurrence:
\[ T(\mu)\ =\ \max_{\substack{\mu_1,\mu_2,\ldots\le t:\\ \mu_1+\mu_2+\cdots\le \mu}} 
\left( \sum_i T(\mu_i) + O(\mu + (\dCross^{2-\delta}/t^{\delta/2} + \dCross/\sqrt{t})\mu\log^{O(1)}\mu) \right).
\]
Let $b$ be a parameter to be chosen later.
Setting $t=\max\{ \log^c \mu,\, b^c\}$ for a sufficiently large constant $c$ (depending on $\delta$)
gives 
\[ T(\mu)\ =\ \max_{\substack{\mu_1,\mu_2,\ldots\le \max\{\log^c\mu,\, b^c\}:\\ \mu_1+\mu_2+\cdots\le \mu}}
\left( \sum_i T(\mu_i) + O((1+\dCross^2/b^2)\mu) \right).
\]
Combined with the naive bound $T(\mu)\le O(\mu\log\mu)\le O(\mu\log b)$ for the base case when $\mu\le b^c$, the recurrence solves to $T(\mu)=O((1+\dCross^2/b^2)\mu\log^*\mu + \mu\log b)$.

If we know $\dCross$, we can simply set $b=\dCross$.
If not, we can ``guess'' its value using an increasing sequence.
More precisely, for $j=\lceil\log(\log^*\mu)\rceil,\lceil\log(\log^*\mu)\rceil+1,\ldots$, we set $b= 2^{2^j}$ and run the algorithm with
running time capped at $O(\mu 2^j)$.
The algorithm succeeds at the latest when $j$ reaches $\max\{\lceil\log(\log^*\mu)\rceil,\lceil\log\log \dCross\rceil\}$, and the total running time is dominated by that of the last iteration, which is $O(\mu\log^*\mu + \mu\log\dCross)$.
%
\end{proof}

The same approach can be applied to EMSTs:

\begin{theorem}\label{thm:cross:MST}
Given a triangulation $G$ of a set $P$ of $n$ points in the plane,
the Euclidean minimum spanning tree $\EMST(P)$ and the Delaunay triangulation $\DT(P)$ can be constructed
in $O(n(\log^* n)^2 + n\log \dCrossMST)$ expected time,
where $\dCrossMST$ is the (unknown) maximum number of edges of $\EMST(P)$ that an edge of $G$ may 
cross.
\end{theorem}
\begin{proof}
We modify the algorithm in the proof of Theorem~\ref{thm:cross}, replacing $\DT(\cdot)$ with $\EMST(\cdot)$ and
``triangulations'' with ``spanning trees''.  All the coefficients 3 are replaced with 1.  In step~6, we invoke
Corollary~\ref{cor:EMST} instead (which causes the extra $\log^*n$ factor).  When guessing $\dCrossMST$,
if we cap the running time at $C\mu 2^j$ for a sufficiently large constant $C$, the probability that $j$ needs to
reach $\max\{\lceil\log(\log^*\mu)\rceil,\lceil\log\log \dCrossMST\rceil+s\}$ for a given $s$ is $O(1/C)^s$, 
so the expected total running time is 
\[ 
\sum_{s\ge 0} O(1/C)^s\cdot C(\mu\log^*\mu + 2^s\mu\log\dCrossMST) = O(\mu\log^*\mu + \mu\log\dCrossMST).
\]

After computing $\EMST(P)$, we can obtain $\DT(P)$ e.g.\ by Devillers' algorithm~\cite{Devillers92} in $O(n\log^*n)$ additional expected time.
\end{proof}

We can also obtain results involving multiple parameters; for example, it is not difficult to adapt the proof of Theorem~\ref{thm:cross} to get a Delaunay triangulation algorithm with running time $O(n\log^*n + \Dwrong \log \dCross)$.


\section{An $O(n\log\log n + n\log(1/\rho))$ algorithm in the probabilistic model}

Using our $\dCross$-sensitive algorithm, we can immediately obtain an efficient algorithm for Delaunay triangulations in the probabilistic model.

\begin{lemma}
Given a set $P$ of $n$ points in the plane and $\rho\in (0,1)$,
let $R$ be a subset of the edges of $\DT(P)$ where each edge is selected independently with
probability $\rho$.
Let $G$ be any triangulation of $P$ that contains $R$.
Let $\dCross$ be the maximum number of edges of $\DT(P)$ that an edge of $G$ may 
cross.
Then $\dCross\le O((1/\rho)\log n)$ with high probability.
\end{lemma}
\begin{proof}
This follows from a standard sampling-based hitting set or epsilon-net argument:
None of the edges of $G$ cross $R$.
For a fixed $p,q\in P$ such that $pq$ crosses more than $(c/\rho)\ln n$ edges of $\DT(P)$,
the probability that $pq$ does not cross $R$ is at most
$(1-\rho)^{(c/\rho)\ln n}\le (e^{-\rho})^{(c/\rho)\ln n} = n^{-c}$.
By a union bound, the probability that $\dCross > (c/\rho)\ln n$ is at most 
$n^{-c+2}$.
\end{proof}

\begin{corollary}
Given a set $P$ of $n$ points in the plane and $\rho\in (0,1)$,
let $R$ be a subset of the edges of $\DT(P)$ where each edge is selected independently with
probability $\rho$.
Let $G$ be any triangulation of $P$ that contains $R$.
Given $G$,
the Delaunay triangulation $\DT(P)$ can be constructed
in $O(n\log\log n + n\log (1/\rho))$ time with high probability.
\end{corollary}
\begin{proof}
We just apply \cref{thm:cross} in combination with the preceding lemma.
\end{proof}

A similar approach works for EMST:

\begin{lemma}
Given a set $P$ of $n$ points in the plane and $\rho\in (0,1)$,
let $R$ be a subset of the edges of $\EMST(P)$ where each edge is selected independently with
probability $\rho$.
Let $T$ be any spanning tree that contains $R$ and does not have any crossings.
Let $\dCrossMST$ be the maximum number of edges of $\EMST(P)$ that an edge of $T$ may 
cross.
Then $\dCrossMST\le O((1/\rho)\log n)$ with high probability.
\end{lemma}

\begin{corollary}
Given a set $P$ of $n$ points in the plane and $\rho\in (0,1)$,
let $R$ be a subset of the edges of $\EMST(P)$ where each edge is selected independently with
probability $\rho$.
Let $T$ be any spanning tree that contains $R$ and does not have any crossings.
Given $T$,
the Euclidean minimum spanning tree $\EMST(P)$ and the Delaunay triangulation $\DT(P)$ can be constructed
in $O(n\log\log n + n\log (1/\rho))$ time with high probability.
\end{corollary}
\begin{proof}
We can compute a triangulation $G$ of $P$ from $T$, e.g., by Chazelle's algorithm in linear time~\cite{Chazelle91}
or Seidel's algorithm in $O(n\log^*n)$ expected time~\cite{Seidel91}.
We can then apply \cref{thm:cross:MST} in combination with the preceding lemma.
\end{proof}

A similar result holds in a random flip model:

\begin{lemma}
Given a set $P$ of $n$ points in the plane and $\rho\in (0,1)$, consider the following process:
\begin{quote}
($\ast$) Initialize $G=\DT(P)$.  Let $m$ be the number of edges in $G$.
For $i=1,\ldots,(1-\rho)m$: pick a random edge $e_i$ from $G$, and if the two triangles
incident to $e_i$ in $G$ forms a convex quadrilateral, flip $e_i$ in $G$.
\end{quote}
At the end of the process, let $\dCross$ be the maximum number of edges of $\DT(P)$ that an edge of $G$ may 
cross.
Then $\dCross\le O((1/\rho)\log n)$ with high probability.
\end{lemma}
\begin{proof}
Let $t=(1-\rho)m$.  For the analysis,
it is helpful to assign a distinct label $\ell(e)\in\{1,\ldots,m\}$ to each edge $e$ in $G$.
When an edge is flipped, we give the new edge the same label as the old edge, to maintain the invariant.
Then $\ell(e_1),\ldots,\ell(e_t)$ are independent, uniformly distributed random variables from $\{1,\ldots,m\}$.
Consider a fixed $p,q\in P$ such that $pq$ crosses more than $b := (c/\rho)\ln n$ edges of $\DT(P)$.
Let $B$ be the labels of a fixed subset of $b$ of these edges.
The probability that $pq$ is in the final $G$ is at most the probability
that $\{\ell(e_1),\ldots,\ell(e_t)\}$ contains $B$, which is
equal to ${t\choose b}\cdot \frac{b!}{m^b} = \frac{t!}{(t-b)!m^b} \le \frac{t^b}{m^b} = (1-\rho)^b \le e^{-\rho b} = n^{-c}$.
By a union bound, the probability that $\dCross > (c/\rho)\ln n$ is at most 
$n^{-c+2}$.
\end{proof}

\begin{corollary}
Given a set $P$ of $n$ points in the plane and $\rho\in (0,1)$,
let $G$ be a triangulation of $P$ generated by the process ($\ast$) above.
Given $G$,
the Delaunay triangulation $\DT(P)$ can be constructed
in $O(n\log\log n + n\log (1/\rho))$ time with high probability.
\end{corollary}

\section{An $O(n\log^*n + n\log\dVio)$ algorithm}

Using our $\dCross$-sensitive algorithm, we can also immediately obtain a $\dVio$-sensitive algorithm for Delaunay triangulations, after proving that $\dCross$ is always upper-bounded by $O(\dVio^2)$.

\begin{lemma}\label{lem:circle0}
Given a triangulation $G$ of a set $P$ of $n$ points in the plane,
let $\dVio$ be the maximum number of points of $P$ that the circumcircle of a triangle of $G$ may contain.
Consider a triangle $\triangle qq''q'''$ in $G$.
The number of edges of $\DT(P)$ incident to $q$ that cross $q''q'''$ is at most $\dVio$.
\end{lemma}
\begin{proof}
Suppose $qq'$ is an edge of $\DT(P)$ that crosses $q''q'''$.  Then
$q$ must be inside the circumcircle of $\triangle qq''q'''$ (because otherwise, $qq'$ and $q''q'''$ would both be edges of $\DT(\{q,q',q'',q'''\})$, contradicting planarity of  Delaunay triangulations).
So, there are at most $\dVio$ choices of $q'$.
\end{proof}

\begin{lemma}\label{lem:vio}
Given a triangulation $G$ of a set $P$ of $n$ points in the plane,
let $\dCross$ be the maximum number of edges of $\DT(P)$ that an edge of $G$ may 
cross, and 
let $\dVio$ be the maximum number of points of $P$ that the circumcircle of a triangle of $G$ may contain.
Then $\dCross\le O(\dVio^2)$.
\end{lemma}
\begin{proof}
Consider an edge $pp'$ of $G$.  Let $\triangle pp'p''$ be a triangle of $G$ incident to this edge.
Let $C$ be the circumcircle of $pp'p''$. 
Let $H$ be the subgraph of $G$ induced by the points inside $C$.

Suppose $qq'$ is an edge of $\DT(P)$ that crosses $pp'$. Then
at least one of $q$ and $q'$ must be inside $C$ (because otherwise,
$pp'$ and $qq'$ would both be edges of $\DT(\{p,p',q,q'\})$, contradicting planarity of  Delaunay triangulations).
W.l.o.g., suppose $q$ is inside $C$.
Let $\triangle qq''q'''$ be the triangle in $G$ incident to $q$ such that $q'$ lies in the wedge
between $\RAY{qq''}$ and $\RAY{qq'''}$.  See Figure~\ref{fig:circle}.
Note that $\triangle qq''q'''$ does not intersect $pp'$ (since $G$ is planar).
Thus, $qq'$ must intersect $q''q'''$.  By Lemma~\ref{lem:circle0}, there are at most $\dVio$ choices of $q'$ for each fixed $\triangle qq''q'''$.

Furthermore, we must have one of the following: (i)~$q''$ or $q'''$ is inside $C$, or (ii)~$p$ and $p'$ lies in the wedge
between $\RAY{qq''}$ and $\RAY{qq'''}$.
In case (i), there are at most $2\cdot\DEG_H(q)$ choices of $\triangle qq''q'''$ for each fixed $q$.
In case (ii), there is only one choice of $\triangle qq''q'''$ per $q$.
Thus, the total number of choices of $\triangle qq''q'''$ is $O(\sum_q 2\DEG_H(q)+1))$,
which is $O(\dVio)$ (since $H$ is a planar graph with at most $\dVio$ vertices).
The total number of choices of $qq'$ is $O(\dVio^2)$.
\end{proof}

\begin{figure}
\centering
\includegraphics[scale=0.8]{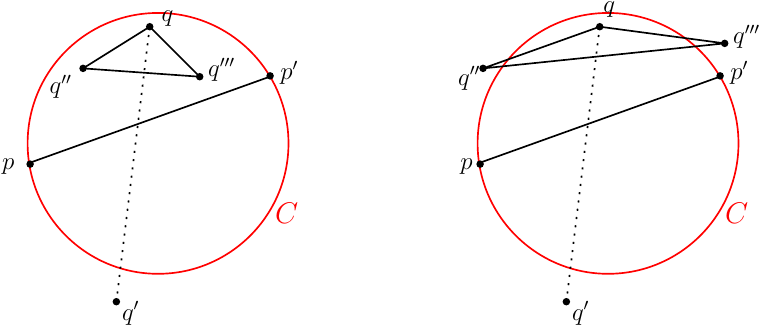}
\caption{Proof of Lemma~\ref{lem:vio}.  The left example is in case (i); the right example is in case (ii).}\label{fig:circle}
\end{figure}

\begin{corollary}
Given a triangulation $G$ of a set $P$ of $n$ points in the plane,
the Delaunay triangulation $\DT(P)$ can be constructed
in $O(n\log^* n + n\log \dVio)$ deterministic time,
where  $\dVio$ is the (unknown) maximum number of points of $P$ that the circumcircle of a triangle of $G$ may contain.
\end{corollary}

Alternatively, it is easier to prove the inequality $\dCrossMST \le O(\dVio)$ (since the EMST has constant maximum degree); we can then apply \cref{thm:cross:MST} to obtain the almost same result, except with one more $\log^*n$ factor.

\section*{Funding}

This research was funded in part by the Slovenian Research and Innovation Agency 
(P1-0297, J1-2452, N1-0218, N1-0285) and in part by the European 
Union (ERC, KARST, project number 101071836).
Views and opinions expressed are
however those of the authors only and do not necessarily reflect those
of the European Union or the European Research Council.  Neither the
European Union nor the granting authority can be held responsible for
them.

\section*{Acknowledgements}

The authors would like to thank Kostas Tsakalidis for suggesting the area of algorithmic problems with predictions.

{\small
\bibliography{main,more_refs}
}

\newpage
\begin{appendix}
\markboth{Appendix}{Appendix}

\section{A deterministic $O(n+\Dwrong\log^4n)$-time algorithm}\label{app:det}


In this appendix, we describe a different deterministic $\Dwrong$-sensitive algorithm for computing the Delaunay triangulation.
Although it is a logarithmic factor slower than the algorithm in \cref{sec:DT:det}, 
the approach does not use planar graph separators and may be of independent interest (it uses certain dynamic geometric data structures, and the running time could potentially be better if those data structures could be improved).

It will be convenient to introduce here the more general notion of a \emph{plane straight-line graph}, shortened to \emph{PSLG}, i.e., a planar embedding of a graph where all edges are drawn as straight-line edges. Thus, a triangulation of a planar point set $P$ is a 
connected PSLG such that: (i) its vertex set is precisely $P$,
(ii) all its bounded faces are triangles, and
(iii) the unbounded face is the complement of the convex hull $\CH(P)$.
We remark that a triangle with an extra vertex inside,
possibly isolated, is not a triangular face of a PSLG.
The combinatorial complexity of a face is the number of edges plus the 
number of vertices that bound the face. Isolated vertices inside
the face also contribute to the count. 

For any three points $p,q,r$, let $C(p,q,r)$ be 
the unique circle through $p,q,r$. For a triangle $\triangle$ 
we denote by $C(\triangle)$ its circumcircle.



Let $G$ be a PSLG and let $f$ be a face of $G$.
We denote by $N_G(f)$ the set of neighbor faces of $f$ in the dual
graph of $G$. 
Therefore, a face $f'$ belongs to $N_G(f)$ if and only if 
$f$ and $f'$ share some edge. 
Let $V(N_G(f))$ be the vertex set of the faces of $N_G(f)$.
Note that if a face of $N_G(f)$ contains isolated vertices
in their interior, then those vertices also belong to $V(N_G(f))$.
See Figure~\ref{fig:DT1}.

\begin{figure}[h]\centering
	\includegraphics[scale=1,page=4]{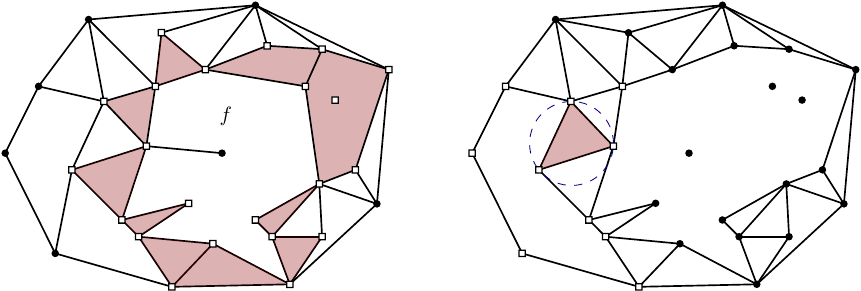}
	\caption{Left: $N_G(f)$ are marked for the face $f$, and the points
		of $V(N_G(f))$ are marked with empty squares.
		The PSLG on the left does not satisfy the conditions of 
		Lemma~\ref{lem:subgraphDT}, the one on the right does.
		To verify the assumptions of Lemma~\ref{lem:subgraphDT} 
		for the shaded triangle on the right we have to check the points marked 
		with empty squares.}
	\label{fig:DT1}
\end{figure}

The standard empty circle property of Delaunay triangulations (given in the introduction) implies that one can test whether a given triangulation is the Delaunay triangulation in linear time. 
We will need the following extension to identify subgraphs
of the Delaunay triangulation.
See Figure~\ref{fig:DT1}, right for the hypothesis.

\begin{lemma}
\label{lem:subgraphDT}
	Let $G$ be a PSLG with vertex set $P$ and such that 
	each edge of $G$ is on the boundary of $\CH(P)$ or 
	adjacent to a triangular face.
	Assume that, for each triangle $\triangle$ of $G$, all the 
	points of $V(N_G(\triangle))$ are outside the circumcircle $C(\triangle)$. 
	Then each edge of $G$ is an edge of the 
	Delaunay triangulation $\DT(P)$.
\end{lemma}
\begin{proof}
	Extend $G$ to a triangulation $T$ and use the algorithm
	that converts $T$ to the Delaunay triangulation using
	edge flips to make edges legal, 
	as described in~\cite[Chapter 9]{BergCKO08}.
	We show that all the edges of $G$ remain legal
	through the whole algorithm.
	
	Assume, for the sake of reaching a contradiction, that $e=p_ip_j$
	is the first edge of $G$ that we flip because it becomes illegal.
	Let $T_e$ be the triangulation we have before the edge flip of 
	$e=p_ip_j$. It cannot be that $e$ belongs to the boundary of $\CH(P)$
	because those edges are never illegal nor flipped.
	Thus, $e$ is not on the boundary of $\CH(P)$.
	By hypothesis, there is a triangle $\triangle p_ip_jp_k$ in $G$. 
	Let $f$ be the other face of $G$ that has $p_ip_j$ on its boundary; 
	$f$ is a bounded face.

	Since no edge of $G$ has been flipped so far, the edges of $\triangle$ and $f$ 
	keep being edges in the current triangulation $T_e$. 
	This means that $f$ is the union of some triangles of $T_e$.
	For the other triangle $\triangle p_ip_jp_l \neq \triangle p_ip_jp_k$ 
	of $T_e$ having $p_ip_j$ as an edge, it must be that $p_l$ is
	a vertex of $f$. However, because no vertex of $f$ is inside
	the circle $C(p_i,p_j,p_k)$, the point $p_l$ is not in the interior
	of $C(p_i,p_j,p_k)$. This means that $e=p_ip_j$ cannot be an illegal edge 
	in $T_e$, which is a contradiction.
\end{proof}

\begin{lemma}
\label{lem:size}
	Let $G$ be a triangulation of $P$ and let $E'$ 
	be a subset of $E(G)$ with $k$ edges. 
	Remove from $G$ the edges of $E'$ and then
	remove also all edges that have bounded,
	non-triangular faces on both sides.
	In total, we remove $O(k)$ edges.
\end{lemma}
\begin{proof}
	The removal of each edge of $E'$ can produce one additional
	non-triangular face and can increase the total number of 
	vertices on non-triangular faces by $O(1)$.
	Thus, after the removal of all the edges of $E'$ we
	have $O(k)$ non-triangular faces that have all together 
	$O(k)$ vertices.
	
	Now, we remove the edges that have bounded,
	non-triangular faces on both sides.
	See Figure~\ref{fig:DT1}, where the PSLG
	on the right is obtained from the PSLG on the left.
	If the edge has two distinct faces on its sides,
	then its removal merges two non-triangular faces,
	but it does not change the number of vertices incident
	to those faces.
	Therefore, this can happen at most $O(k)$ times
	because, if we started with $\ell$ non-triangular faces,
	this operations can be made at most $\ell-1$ times.
	After this, we have $O(k)$ faces with all together
	$O(k)$ vertices.
	
	Now, we remove the edges that have the same face on both sides.
	The removal of each of those edges does not change the 
	number of faces nor vertices
	incident to non-triangular faces, but
	decreases the number of edges incident to non-triangular faces.
	Since the non-triangular faces form a planar
	graph with $O(k)$ vertices, they have $O(k)$ edges and
	therefore these deletions can happen $O(k)$ times.
\end{proof}
		
We will use the following data structures that follow from
Chan~\cite{Chan20a} by using the usual lift to the 
paraboloid $z=x^2+y^2$.
For the second part, we just use point-plane duality in $\RR^3$.
\begin{lemma}
\label{lem:datastructure}
	We can maintain a dynamic set $P$ of $n$ points in the plane 
	under the following operations:
	\begin{itemize}
		\item insertion or deletion of a point takes $O(\log^4 n)$ amortized time;
		\item for a query circle $C$, we can detect whether
			some point of $P$ lies inside the circle $C$ in $O(\log^2 n)$ time.
	\end{itemize}
	We can maintain a dynamic set $\mathcal{C}$ of $n$ 
	circles in the plane under the following operations:
	\begin{itemize}
		\item insertion or deletion of a circle takes $O(\log^4 n)$ amortized time;
		\item for a query point $p$, we can detect whether
			some circle of $\mathcal{C}$ contains $p$ in $O(\log^2 n)$ time.
	\end{itemize}
\end{lemma}

We next show our main result. 

\begin{theorem}
        Let $P$ be a set of $n$ points in the plane and $G$ a PSLG with vertex set $P$.
        The Delaunay triangulation $\DT(P)$ can be constructed in $O(n + \Dwrong \log^4 n)$ time, where $D$ is redefined as  the (unknown) number of edges in $\DT(P)$ that are not in $G$.  
\end{theorem}
\begin{proof}
	The first step is identifying in $O(n+ \Dwrong \log^4 \Dwrong)$ time a subset $E'$ 
	of edges of $G$ with $O(\Dwrong)$ edges such that 
	all edges of $G-E'$ are from the Delaunay triangulation $\DT(P)$. 
	The subset $E'$ may contain also some edges of $\DT(P)$; 
	thus, we compute a superset of the ``wrong'' edges. 
	The criteria to identify $E'$ is that $G-E'$ should satisfy 
	the assumptions of Lemma~\ref{lem:subgraphDT}.
    There are two basic ideas that are being used for this: 
	\begin{itemize}
	\item Test circumcircles of triangular faces to see whether they contain 
		some point of a neighbor face. The same triangle may have to be checked 
		multiple times because its neighbor faces may change, and we employ 
		data structures to do this efficiently.
	\item If an edge is on the boundary of two non-triangular faces,
		we have no triangle to test the edge. In this case we just add
		the edge to $E'$, but this can happen $O(\Dwrong)$ times because of 
		Lemma~\ref{lem:size}.
	\end{itemize}
	
	We now explain the details.  
	For any subgraph $G'$ of $G$, let $F_{\ge 4}(G')$ be 
	the faces of $G'$ that are non-triangular and bounded.
	Thus, the complement of $\CH(P)$ is never in $F_{\ge 4}(G')$.
	If $G$ is a triangulation, then $F_{\ge 4}(G)$ is the empty set.
	At the start, $F_{\ge 4}(G)$ has $O(\Dwrong)$ vertices; even if $G$
	would be a subgraph of $\DT(P)$, we still have to add $\Dwrong$ edges to $G$
	to obtain $\DT(P)$, and each non-triangular face with $\ell>3$ vertices
	will need $\ell-3>0$ edges to get triangulated.
	We do not maintain $F_{\ge 4}(G)$ explicitly,
	but it is useful to have a notation for it.
	
	We start with $E'=\emptyset$ and the PSLG $G'=G$; we maintain
	the invariant that $G=G'+E'$. Thus, $E'$ is the set of edges that
	are being removed.		

	The triangles of $G'$ are marked as \emph{to-be-checked} or \emph{already-checked}.
	Initially we mark all triangles of $G$ as \emph{to-be-checked}.
	(If $\CH(P)$ is a triangle, the outer face is marked as 
	\emph{already-checked} and remains so permanently.)

	Let $Q$ be the set of points that lie in the faces of $F_{\ge 4}(G')$. 	
	We maintain $Q$ in the first data structure	of Lemma~\ref{lem:datastructure}. 
	At the start $Q$ has $O(\Dwrong)$ elements.
	We also maintain the set $\tau$ of triangles that share an edge
	with a face of $F_{\ge 4}(G')$ and are marked as \emph{already-checked}. 
	For each triangle of $\tau$,
	its circumcircle is stored in the second data structure of 
	Lemma~\ref{lem:datastructure}. 
	At the start $\tau$ has no elements because all triangular faces are
	marked as \emph{to-be-checked}.
	We maintain the invariant that no point of $Q$ is contained
	in the interior of $C(\triangle)$ for any $\triangle \in \tau$.
	
	We iterate over the triangles $\triangle$ of $G'$ that are marked 
	as \emph{to-be-checked}, and for each of them proceed as follows. 
	We want to decide whether $C(\triangle)$ contains some point of $V(N_{G'}(\triangle))$.
	If some edge $e$ of $\triangle$ is shared with another
	triangle $\triangle'$ of $G'$, we test whether the vertices of $\triangle'$
	lie in interior of the circumcircle $C(\triangle)$.
	We also query the data structure storing $Q$ to see whether
	$C(\triangle)$ contains some point of $Q$. This covers 
	the case when some face $f$ adjacent to $\triangle$ is non triangular
	because $V(f)\subseteq Q$. 	
	With this we have tested whether any point from $V(N_{G'}(\triangle))$
	is contained in $C(\triangle)$; we also tested some additional
	points, namely $Q\setminus V(N_{G'}(\triangle))$.

	If we detected some point in the interior of $C(\triangle)$,
	we delete the three edges of $E(\triangle)$ immediately after 
	its the detection and it is part of the iteration of $\triangle$.
	Note that when we decide to delete $E(\triangle)$,
	we know for sure that at least one of the edges of $\triangle$ is not 
	in $DT(P)$. Thus, in this way we delete at most $3\Dwrong$ edges.

	If the edges of $\triangle$ are not to be deleted, then we 
	mark $\triangle$ as \emph{already-checked}. In this case, if $\triangle$
	is neighbor to some face of $F_{\ge 4}(G')$, we 
	add it to $\tau$ and its circumcircle $C(\triangle)$ is added to
	the corresponding data structure.
	
	Let us consider now what has to be done to delete $E(\triangle)$
	for a triangle $\triangle$. We remove from $G'$ the edges of $E(\triangle)$ 
	that are not on the boundary of $\CH(P)$ and add them $E'$.
	The triangle $\triangle$ now becomes part of a new face, 
	let us call it $f$, that belongs to $F_{\ge 4}(G')$. 
	We restore the invariants doing the following:
	\begin{itemize}
	\item We add to $Q$ the vertices $Q_{\rm new}$ of $f$ that were 
		not in $Q$ before; there are at most $O(1)$ such 
		new points added to $Q$ because, besides the vertices of $\triangle$, 
		only the neighbor triangles of $\triangle$ may contribute one 
		new point to $Q$.
	\item The face $f$ may have neighbor triangles $\tau_{\rm new}$
		marked as \emph{already-checked} that have to be added to $\tau$. 
		Again, there are $O(1)$ such new triangular faces because 
		only the triangular faces
		adjacent to $\triangle$ may have a new triangular face for $\tau$.
		Therefore, we can identify these new triangular faces $\tau_{\rm new}$
		for $\tau$ in $O(1)$ time.
		We also add the circumcircles of the triangles in
		$\tau_{\rm new}$ to the corresponding data structure.
	\item We have to check again each triangle $\triangle$ of $\tau_{\rm new}$
		because their neighbor faces changed and $Q$ contains new points 
		that may be inside $C(\triangle)$.
		For each triangle $\triangle$ of $\tau_{\rm new}$, we query the 
		data structure storing $Q$ to see whether $\triangle$ contains some point
		of $Q$ inside. If it does, we mark $\triangle$ as \emph{to-be-checked};
		it will be tested again and deleted, unless it stops being
		a triangle in $G'$ because of some other deletion.		
	\item The triangles of $\tau$ that are adjacent to $f$ have to be 
		checked again because $V(f)$ contains new points, namely $Q_{\rm new}$. 
		See Figure~\ref{fig:DT2} for an example.
		We do this using the data structure storing $C(\triangle')$ for $\triangle'\in \tau$.
		We query with each point of $Q_{\rm new}$ the data structure 
		multiple times to find 
		\[ 
			\tau' = \{ \triangle'\in \tau\mid \text{$C(\triangle')$ contains 
					some point of $Q_{\rm new}$ in the interior}\}.
		\]
		We delete the triangles of $\tau'$ from $\tau$; this is actually
		done already during the computation of $\tau'$.
		We mark the triangles of $\tau'$ as \emph{to-be-checked}. 
		Each triangle of $\tau'$ will be tested again later and deleted, 
		unless it stops being a triangle in $G'$ because of some other deletion.		
	\end{itemize}
	
	\begin{figure}\centering
	\includegraphics[scale=1.1,page=5]{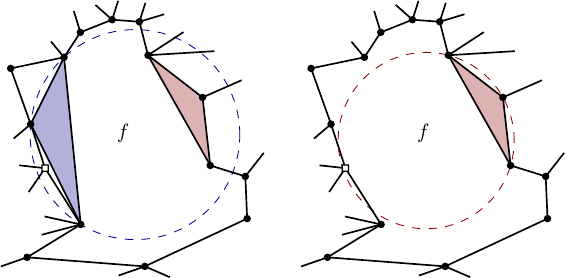}
	\caption{When we check the blue-shaded triangle $\triangle$, we decide to delete $E(\triangle)$ 
		because $C(\triangle)$ contains some point of $V(f)\subseteq Q$.
		The point marked with an empty square becomes a new point of $Q$ 
		that is in the interior of $C(\triangle')$, where $\triangle'$ is red-shaded. 
		Because of this, $\triangle'$ is marked (perhaps again) as \emph{to-be-checked}.}
	\label{fig:DT2}
	\end{figure}

	This finishes for now the description of the computation of $G'$ and $E'$. 
	For each triangle we delete, at least one of its edges is not $\DT(P)$.
	Therefore, we delete at most $3\Dwrong$ edges.
	For each triangle we delete, we add $O(1)$ new points to $Q$ and
	add $O(1)$ triangles to $\tau$. Thus, we make $O(\Dwrong)$ editions in
	the data structures of Lemma~\ref{lem:datastructure} storing $Q$ and 
	the circumcircles of triangles in $\tau$.
	In $Q$ we make no deletions, while in $\tau$ we make
	a deletion when we know that the triangle will be deleted. 
	For each triangle we delete, we mark some triangles 
	as \emph{to-be-checked}, but this happens only when we know
	that it will be deleted. Thus, we mark back at most $O(\Dwrong)$
	triangles from \emph{already-checked} to \emph{to-be-checked}.
	This shows that we process in total $O(n+\Dwrong)$ triangular	faces.

	Queries to the circumcircles $\{ C(\triangle')\mid \triangle'\in \tau\}$ are only made 
	when we delete a triangular face, and for each such deletion 
	we perform $O(1)$ queries with $Q_{\rm new}$; therefore there are 
	$O(\Dwrong)$ such queries to the circumcircles of $\tau$. 
	Queries to $Q$ are are made with circles $C(\triangle)$ for a triangle $\triangle$
	because $\triangle$ is adjacent to a face of $F_{\ge 4}(G')$. 
	Each of those triangles either was adjacent to a face of $F_{\ge 4}(G)$,
	the starting graph $G$, or it is one of the $O(1)$ triangles in $\tau_{\rm new}$ 
	that has to be considered because of a deletion of another triangle. 
	The first case happens $O(\Dwrong)$ times because all faces of $F_{\ge 4}(G)$ 
	have at most $O(\Dwrong)$ edges together. The second case happens $O(\Dwrong)$ times
	because each deletion of a triangle triggers $O(1)$ new queries, and we delete 
	$O(\Dwrong)$ triangles. In total, we are making $O(\Dwrong)$ queries to the data
	structure storing $Q$.
	We conclude, that the queries to the data structures storing $Q$ and 
	$\{ C(\triangle')\mid \triangle'\in \tau\}$ take $O(\Dwrong \log^4 \Dwrong)$ time all together.
	
	We summarize the current state: we have computed in $O(n+\Dwrong \log^4 \Dwrong)$ 
	time a set $E'$ of $O(\Dwrong)$ edges of $G$ such that, for all triangular
	faces $\triangle$ of $G'=G-E'$, the circumcircle $C(\triangle)$ contains no points 
	of $Q=V(F_{\ge 4}(G'))$ and no points of triangles adjacent to $\triangle$.
	Now, we remove from $G'$ all the edges 
	that have on both sides non-triangular faces, and add them to $E'$. 
	Since this does not change $Q=\bigcup_{f\in F_{\ge 4}(G')} V(f)$ nor 
	$\tau$, the condition 
	we had to finish is not altered: the circumcircle of 
	each triangle does not contain in its interior any vertex
	of $V(F_{\ge 4}(G'))$ nor its neighbor triangles.
	Because of Lemma~\ref{lem:size} we have deleted additionally 
	$O(3\Dwrong)=O(\Dwrong)$ edges, which are added to $E'$.
	
	We have computed a set $E'$ of $O(\Dwrong)$ edges from the input
	PSLG such that $G-E'$ satisfies the assumptions
	of Lemma~\ref{lem:subgraphDT}.
	It follows that all edges of $G-E'$ are edges $DT(P)$.
	We then compute the constrained Delaunay triangulation for 
	the faces of $G-E'$ using the algorithm of Chew~\cite{Chew89}; 
	those faces may contain some isolated points 
	or connected components inside. 
	First we spend $O(n)$ time to detect which
	faces of $G-E'$ are not triangles. 
	Since $G-E'$ has $O(\Dwrong)$ non-triangular faces
	with a total of $O(\Dwrong)$ vertices, 
	the computation of the constrained Delaunay triangulation
	for all non-triangular faces together takes $O(\Dwrong \log \Dwrong)$ time.
\end{proof}

%

\end{appendix}

\end{document}